\documentclass[lettersize,journal]{IEEEtran}
\usepackage{amsmath,amssymb,amsfonts,amsthm}
\usepackage{algorithmic}
\usepackage{algorithm}
\usepackage{array}
\usepackage[caption=false,font=normalsize,labelfont=sf,textfont=sf]{subfig}
\usepackage{textcomp}
\usepackage{stfloats}
\usepackage{url}
\usepackage{verbatim}
\usepackage{graphicx}
\usepackage{cite}
\usepackage{xcolor}
\usepackage{booktabs}
\usepackage{multirow}
\usepackage{microtype}
\usepackage{bm}
\usepackage{enumitem}

\newtheorem{theorem}{Theorem}
\newtheorem{proposition}{Proposition}
\newtheorem{corollary}{Corollary}
\newtheorem{definition}{Definition}
\newtheorem{remark}{Remark}
\newtheorem{example}{Example}

\def\BibTeX{{\rm B\kern-.05em{\sc i\kern-.025em b}\kern-.08em
    T\kern-.1667em\lower.7ex\hbox{E}\kern-.125emX}}

\begin{document}

\title{Post-Hoc Trajectory-Risk Certification for Modular LLM-Based Security Agents}

\author{Zhenpeng~Li%
\IEEEcompsocitemizethanks{%
\IEEEcompsocthanksitem Zhenpeng Li is with Guangzhou Health Science College,
No.~248 Guangyuan Middle Road, Guangzhou, Guangdong 510405, China
(e-mail: 2025301001@gzws.edu.cn). Corresponding author: Zhenpeng Li.}}

\maketitle

\begin{abstract}
Autonomous security agents increasingly operate as staged decision
pipelines, e.g., classifying network traffic and then attributing
detected attacks to a specific technique. Split conformal prediction
gives each stage a finite-sample coverage guarantee, but deployment
requires a trajectory-level guarantee across the whole chain, and the
two do not compose for free---especially for stages that are already
independently trained and calibrated and cannot be jointly
recalibrated. Bonferroni allocation gives a valid distribution-free
trajectory bound but is conservative when stage errors are correlated.
We show a natural pairwise-correlation extension of this bound to
three or more stages is invalid---a lower, not upper, bound---and give
a provably valid spanning-tree alternative. We then separate two
questions routinely conflated in practice: whether stages are
dependent at all, and whether a finite audit sample is large enough to
certify that dependence, giving matching upper and information-theoretic
lower sample-complexity bounds for both. We further prove that a
common design pattern---using a coarse category to select a
fine-grained label space---mechanically manufactures near-perfect
measured stage correlation with no learned dependence behind it.

On a two-stage intrusion-detection pipeline (traffic classification,
then attack-variant attribution) across 6 open LLMs and 2
datasets, avoiding this labeling artifact drops measured correlation
from near-1 to a genuine, task-dependent $0$--$0.78$.  A direct audit
of the full trajectory-failure event, requiring simultaneous access to
both stages, becomes $13.7\%$ tighter than Bonferroni once the audit is
scaled to the sample complexity our theory requires---worse than
Bonferroni at an under-sized audit sample, consistent with the derived
threshold; our modular certificate, which composes only per-stage
certificates and a pairwise overlap bound, recovers a smaller but
strictly positive certified gain (0.6\% on average) at the same scale,
quantifying the cost of not having joint access to both stages. A
same-model/cross-model/permuted-pairing test shows the residual
dependence reflects shared sample difficulty, not shared model
representations. Trajectory coverage holds at the nominal $\alpha=0.10$
target on average across all 12 tested configurations
($92.7\%\pm2.4\%$), and deploying one dataset's calibrated classifier
on another's real traffic drives single-step miscoverage to $100\%$
even when raw classification accuracy remains as high as $78\%$:
calibrated confidence, not accuracy, is what distribution shift
destroys.
\end{abstract}


\section{Introduction}
\label{sec:intro}

Security Operations Centers (SOCs) face alert volumes that consistently
exceed human analyst capacity.  LLM-based autonomous
agents---exemplified by Microsoft Copilot for Security and Google
Security AI Workbench---are emerging as first responders that triage,
classify, and act on network intrusions without human intervention,
rarely in a single decision: a typical autonomous IDS pipeline chains
\emph{traffic classification} (benign vs.\ attack category) with
\emph{threat attribution} (mapping to a specific MITRE ATT\&CK
technique), and a production system may chain a third step---automated
response (block, quarantine, escalate)---onto both.  If the first step
misclassifies a denial-of-service attack as benign, the second step
never gets the chance to attribute it correctly.

Split conformal prediction (CP)~\cite{vovk2005algorithmic,
papadopoulos2002inductive} provides \emph{distribution-free,
finite-sample} guarantees for individual classification decisions:
given a user-specified miscoverage budget $\alpha$, the conformal
prediction set $C(x)$ satisfies $\Pr[y \notin C(x)] \leq \alpha$---a
guarantee on the \emph{set}, not on the model's raw softmax confidence,
which is well known to be poorly calibrated on its
own~\cite{guo2017calibration}.  This single-step guarantee underlies
the IDS abstention problem.

However, per-step guarantees do not compose into trajectory-level
guarantees without additional analysis.  If each step independently
satisfies $\Pr[y_k \notin C_k(x)] \leq \alpha$, the probability
that \emph{at least one step} in a $K$-step pipeline fails to cover
the true label can be as large as $K\alpha$.  The standard remedy
is Bonferroni correction: allocate $\alpha/K$ per step, ensuring the
trajectory-level error rate is at most $\alpha$.  But this comes at
a steep cost.  For a 3-step pipeline with $\alpha = 0.10$, each step
receives only $\alpha/3 \approx 0.033$, producing prediction sets so
large that the agent must escalate (abstain) on a far greater fraction
of inputs than necessary.  In SOC environments where every unnecessary
escalation consumes analyst time, this conservatism directly undermines
the value of automation.

The central question of this paper is: \emph{how should per-step risk
certificates compose into a valid and informative trajectory-level
guarantee, and what does this composition reveal in LLM-based security
pipelines?}  The first half is a general question about any chain of
calibrated decisions; the second is the concrete instance we study in
depth.  We treat the general question as the frame and the two-step
IDS pipeline as its primary---and, in this paper, only---worked
instantiation.  The post-hoc certification framework and the $K=2$
identity (Sections~\ref{subsec:mrcc}, \ref{subsec:certifiability}) are
directly exercised on this pipeline in Section~\ref{sec:results}; the
general-$K$ results ($K\geq3$ non-identifiability, the spanning-tree
bound's multi-edge tree selection, and the dependence-duality
redundant-system comparison) characterize the extension beyond $K=2$
and its limits, and are not tested on a real multi-step pipeline in
this paper.

The composition question has two parts.  The distribution-free part of
the answer is the Bonferroni trajectory bound, which holds unconditionally.
Beyond that guarantee, we show that a correlation-aware
inclusion--exclusion estimate can be tighter in the studied IDS
pipeline, that a naive extension of this estimate to more than two
steps is invalid, and that a distribution-free alternative exists.
Whether the two-step correlation-aware estimate helps in a given
deployment, and how much, depends on the task: across 2 datasets and 6
LLMs, using a genuinely independent attack-variant classifier for
Step~2 (Section~\ref{sec:setup}), the mean computable inter-step Pearson
correlation is $\bar{\rho} = 0.30$, ranging from $\approx 0$ on
CIC-IDS-2018/DoS to $0.58$ on RT-IoT2022/Probe.  Hard samples tend to be
hard at both stages on RT-IoT2022, less so on CIC-IDS-2018.  We show
this heterogeneity is not noise but is governed by each task's
underlying joint-failure rate $q_{12}$, and by whether the joint audit
sample is large enough to certify it (Section~\ref{subsec:certifiability}).

\textbf{Contributions.}  This paper makes the following contributions:

\begin{enumerate}[leftmargin=*]
\item \textbf{A general composition framework for marginally
risk-controlled chains.}  We formalize trajectory-level coverage for
any \emph{marginally risk-controlled chain}---a sequence of decisions
each with a per-step risk certificate---as an abstraction that
subsumes conformal prediction, conformal risk control, and other
calibrated decision rules (Section~\ref{subsec:mrcc}).  Bonferroni
composition and a stronger multiplicative composition under a
survival-conditional guarantee follow as special cases of a single
theorem; we instantiate and evaluate this framework only on split CP
in our experiments.

\item \textbf{Post-hoc certification of dependence gain for frozen
modular chains.}  For a chain whose stages cannot be jointly
recalibrated, we separate a chain's true dependence gain into a
\emph{structural} component (what a pairwise-only audit could reveal
with infinite data) and a \emph{statistical} component (what a finite
audit actually proves), via matching \emph{oracle} and
\emph{certifiable} gains (Definition~\ref{def:certifiable_gain},
Theorem~\ref{thm:gain_decomposition}) and a post-hoc certificate valid
under data-dependent tree selection (Theorem~\ref{thm:posthoc_certificate}).
We give matching upper and information-theoretic lower audit-size
bounds for certifying \emph{any} positive gain
(Theorems~\ref{thm:positive_gain_probability}--\ref{thm:positive_gain_lower_bound}),
distinct from the harder question of beating the nominal Bonferroni
threshold (Section~\ref{subsec:certifiability})---this framework, not
the inclusion--exclusion identity below, is what explains why our own
certified bound reverses sign with audit sample size
(Section~\ref{subsec:hunter_empirical}).

\item \textbf{Correcting the naive $K>2$ extension, and a valid
replacement.}  We show that the natural pairwise-correlation extension
of the $K=2$ inclusion--exclusion identity is a \emph{lower} bound, not
an upper bound, for $K>2$; derive the exact threshold at which it
spuriously degenerates to zero (Theorem~\ref{thm:pairwise_degeneracy});
prove that pairwise information---correlations or exact overlaps---cannot
determine trajectory risk once $K \geq 3$
(Theorem~\ref{thm:pairwise_nonidentifiability}); and give a
distribution-free, finite-sample-certifiable upper bound that remains
valid at any $K$ (the Spanning-Tree Pairwise Bound,
Theorem~\ref{thm:hunter_trajectory}).

\item \textbf{Why positive dependence is an asset here, not a
liability.}  We prove that positive association helps a union-type
trajectory failure criterion while it would hurt an intersection-type
redundant-system criterion (Theorem~\ref{thm:dependence_duality}),
resolving an apparent tension with standard reliability-engineering
intuition about correlated component failures.

\item \textbf{Label-induced coupling and an empirical mechanism test.}
We prove that a deterministic (or near-deterministic) coarse-to-fine
label mapping mechanically nests the two stages' failure events,
producing near-maximal correlation with no learned behavior involved
(Theorem~\ref{thm:label_induced_coupling}); a genuinely independent
attack-variant classification task instead exhibits heterogeneous,
task-dependent coupling ($\bar{\rho} = 0.30$).  A
same-model/cross-model/permuted pairing experiment attributes the
residual coupling to shared per-sample difficulty rather than a
same-model representation-sharing increment
(Section~\ref{subsec:hunter_empirical}).

\item \textbf{Comprehensive evaluation.}  Across 2 datasets, 6 LLMs,
3 $\alpha$ levels, and 5 random seeds (36 configurations), we verify
that Bonferroni trajectory coverage holds at the marginal level, report
the full oracle/certifiable/marginal decomposition of the dependence
gain (not only the final certified value), and report the tested
$\alpha$-budget allocations.

\item \textbf{Exchangeability failure under real cross-dataset shift.}
Deploying a stage's own fine-tuned checkpoint on a different dataset's
real traffic drives single-step miscoverage to 100\% in all 12 tested
cells, even at 78\% top-1 accuracy---calibrated confidence, not raw
accuracy, is what exchangeability violation destroys, though the
resulting empty prediction sets (rate $0.96$--$1.00$) are a stark,
label-free symptom available at inference time (Section~\ref{subsec:drift}).
\end{enumerate}

\section{Related Work}
\label{sec:related}

\subsection{Conformal Prediction for Classification}

Conformal prediction~\cite{vovk2005algorithmic} provides distribution-free
coverage guarantees for prediction sets.  Split (inductive) conformal
prediction~\cite{papadopoulos2002inductive} computes nonconformity scores
on a held-out calibration set and selects a threshold $\hat{q}$ such that
the prediction set $C(x) = \{y : s(x,y) \leq \hat{q}\}$ satisfies
$\Pr[y \notin C(x)] \leq \alpha$.  Conformal risk control
(CRC)~\cite{angelopoulos2022conformal} generalizes this to arbitrary
monotone loss functions.  These methods guarantee coverage for
\emph{individual} predictions; our work addresses the orthogonal problem
of guaranteeing coverage across \emph{sequences} of predictions.

\subsection{Multiple Testing and Bonferroni Corrections}

The Bonferroni correction provides a valid union bound under arbitrary
dependence.  The \v{S}id\'{a}k correction~\cite{sidak1967rectangular}
is tighter but, unlike Bonferroni, is exact only under independence and
valid more generally only under specific positive-dependence
structures~\cite{sidak1967rectangular}; we therefore build on
Bonferroni, not \v{S}id\'{a}k, as the distribution-free baseline
throughout this paper.  Both corrections are well-known to be
conservative when errors are positively
correlated~\cite{hochberg1988sharper}.  Holm~\cite{holm1979simple}
and Benjamini--Hochberg~\cite{benjamini1995controlling} offer less
conservative alternatives for independent or positively dependent
test statistics.  Our setting differs from classical multiple testing:
we seek to control the probability of \emph{any} miscoverage in a
fixed-length pipeline, not the family-wise error rate across many
hypotheses.  The key distinction is that pipeline steps are
\emph{structurally coupled} through shared inputs and LLM
representations, producing the positive inter-step correlation that
our inclusion--exclusion analysis measures.

\subsection{Sequential and Multi-Step Conformal Prediction}

Gibbs and Cand\`{e}s~\cite{gibbs2021adaptive} introduce Adaptive
Conformal Inference (ACI) for online settings where exchangeability
is violated over time.  Barber et al.~\cite{barber2023conformal}
extend CP beyond exchangeability via weighted quantiles.  These
methods address \emph{temporal} non-exchangeability (distribution
drift across time steps) rather than \emph{compositional}
non-independence (coverage across pipeline stages within a single
input).  Bates et al.~\cite{bates2021distribution} develop
risk-controlling prediction sets with probably approximately correct
(PAC)-Bayes bounds that can, in principle, handle compound losses, but
require additional held-out data for posterior optimization.  Our
approach is complementary: it uses observed inter-step error coupling to
measure how loose the Bonferroni trajectory bound is in the studied
pipeline.

Most directly related is PASC~\cite{pasc2026}, which reduces
multi-stage joint coverage to a single scalar conformal problem on the
\emph{joint maximum nonconformity score} across all $K$ stages, giving
a distribution-free joint-coverage guarantee tight to $1/(n+1)$.  PASC
and our framework target different deployment regimes rather than
competing on the same problem.  PASC requires simultaneous access to
every stage's nonconformity score on a common calibration set, so that
the joint-maximum score can be computed and calibrated as a single
quantile; this is the right tool whenever the pipeline can be jointly
recalibrated end-to-end.  Our post-hoc modular certificate
(Section~\ref{subsec:certifiability}) instead targets stages that have
already been independently trained and calibrated, expose only
per-stage certificates (or binary pass/fail indicators) rather than
raw nonconformity scores, and so cannot be jointly recalibrated even
though no retraining is required---joint recalibration is a
post-hoc statistical step, but it still needs simultaneous access to
every stage's raw score on a shared calibration set, which the setting
we study does not provide.  This is the common case when stages come
from different vendors, different release cycles, or heterogeneous
score spaces that are not directly comparable, or when only a
pass/fail audit log is retained rather than raw scores.
Where PASC asks ``what is the tightest joint prediction set achievable
if we control the whole pipeline,'' we ask ``what can be certified
about a pipeline we do not control, from audit data alone, and how much
that certification costs.''  When joint recalibration is available,
PASC's construction should be preferred; our contribution is the
regime where it is not.  Non-exchangeable conformal risk
control~\cite{farinhas2024nonexchangeable} relaxes exchangeability
itself, complementary to our Section~\ref{subsec:drift} finding that a
fixed threshold gives no certified warning under exchangeability
violation; class-conditional conformal prediction with many
classes~\cite{ding2023class} addresses per-class validity within one
classifier rather than the cross-stage composition we study.  Our
audit-size requirements connect to sequential
testing~\cite{wald1945sequential} and time-uniform confidence
sequences~\cite{howard2021timeuniform,khosravi2026csa}, which would let
the audit sample size itself be data-dependent, an alternative to our
fixed-$n$ construction.  Outside NLP, audited post-hoc verification of
a frozen, deployed system under distribution shift has been studied for
power-grid contingency screening~\cite{manoharan2026audited}, suggesting
this certification regime is not specific to language-model pipelines.

\subsection{LLM-Based Security Agents}

LLM-based autonomous agents for cybersecurity have attracted growing
attention~\cite{motlagh2024llm,xu2024autoattacker,xu2024large}.
Microsoft Copilot for Security and Google
Security AI Workbench represent production
deployments.  Prior work on the safety of autonomous security decisions
has largely treated each decision in isolation---either its adversarial
robustness or its single-step conformal guarantee, e.g.\ conformal
prediction for online intrusion-detection models under concept
drift~\cite{escudero2025conformal} or inductive conformal anomaly
detection for anomalous sub-trajectories~\cite{laxhammar2015inductive},
both single-stage.  Trajectory-level guarantees for frozen,
independently-calibrated multi-step security pipelines
specifically---as opposed to jointly recalibrated pipelines, which
PASC~\cite{pasc2026} addresses---have not been previously addressed in
this domain.

\section{Problem Formulation}
\label{sec:formulation}

\subsection{Multi-Step Security Agent Pipeline}

We model an autonomous IDS agent as a $K$-step pipeline operating on
network traffic input $x \in \mathcal{X}$.  At each step $k \in
\{1, \ldots, K\}$, the agent produces a prediction $\hat{y}_k$ for
the true label $y_k$:

\begin{itemize}[leftmargin=*]
\item \textbf{Step 1---Traffic Classification.}  Given raw network
flow features $x$, the LLM classifies the traffic into one of $L_1$
categories (e.g., Normal, DoS, Probe, Exploit).
\item \textbf{Step 2---Threat Attribution.}  Given $x$ and the Step~1
output, the LLM maps detected attacks to MITRE ATT\&CK
techniques (e.g., T1498 Network Denial of
Service, T1110 Brute Force), producing one of $L_2$ attributions.
\item \textbf{Step $\bm{k > 2}$---Response or Further Analysis.}
Additional steps (e.g., response recommendation, severity scoring)
follow the same structure.
\end{itemize}

Each step uses a conformal predictor to quantify uncertainty.
Specifically, let $s_k(x, y)$ denote the nonconformity score at step $k$
for input $x$ and candidate label $y$.  Given a calibration set
$\mathcal{D}_{\text{cal}} = \{(x_i, y_{i,1}, \ldots, y_{i,K})\}_{i=1}^n$,
the conformal prediction set at step $k$ is
(Eq.~\eqref{eq:cp_set}):
\begin{equation}
\label{eq:cp_set}
C_k(x; \alpha_k) = \{y : s_k(x, y) \leq \hat{q}_k(\alpha_k)\},
\end{equation}
where $\hat{q}_k(\alpha_k)$ is the $\lceil (1-\alpha_k)(n+1) \rceil / n$
quantile of the calibration nonconformity scores at step $k$.

\subsection{Trajectory Coverage}

\begin{definition}[Trajectory Coverage]
\label{def:tc}
A $K$-step trajectory $\tau = (x, y_1, \ldots, y_K)$ is \emph{covered}
if the true label at every step lies within the corresponding conformal
prediction set (Eq.~\eqref{eq:tc}):
\begin{equation}
\label{eq:tc}
\text{covered}(\tau) \iff \bigwedge_{k=1}^{K} \left[ y_k \in C_k(x; \alpha_k) \right].
\end{equation}
The \emph{trajectory coverage rate} is:
\begin{equation}
\text{TC} = \Pr\!\left[\bigwedge_{k=1}^{K} y_k \in C_k(x; \alpha_k)\right].
\end{equation}
\end{definition}

\begin{definition}[Trajectory Miscoverage Rate]
\label{def:tar}
The trajectory miscoverage rate is the probability that at least one
step fails to cover the true label:
\begin{equation}
\text{TMR} = 1 - \text{TC} = \Pr\!\left[\bigvee_{k=1}^{K} y_k \notin C_k(x; \alpha_k)\right].
\end{equation}
\end{definition}

TMR is a statistical risk quantity, not a directly observable
operational one: evaluating $\text{covered}(\tau)$ requires the true
labels $y_1, \ldots, y_K$, which are unavailable at inference time.
TMR is therefore what a trajectory-level conformal guarantee controls,
not what a deployed agent can monitor online.  A deployment separately
needs an \emph{observable} escalation rule, e.g.\ $A_k = \{|C_k(X)| \neq 1\}$
(escalate whenever a step's prediction set is not a singleton) or a
more general policy $A_k = \{\pi_k(C_k(X)) = \texttt{escalate}\}$; the
resulting trajectory escalation rate $\text{TER} = \Pr[\bigcup_k A_k]$
is generally \emph{not} equal to TMR, since prediction-set size and
miscoverage are related but distinct random quantities.  We report TMR
throughout this paper because it is what our theorems control; we do
not compute TER, and readers should not conflate the two.

\subsection{The $\alpha$-Allocation Problem}

Given a total trajectory-level miscoverage budget $\alpha$, the
deployment operator must choose per-step budgets $\alpha_1, \ldots,
\alpha_K$ such that:
\begin{equation}
\label{eq:alpha_budget}
\text{TMR} \leq \alpha.
\end{equation}

The goal is to satisfy~\eqref{eq:alpha_budget} while minimizing TMR,
the statistical risk that the trajectory-level guarantee controls (TMR
is distinct from the observable escalation rate TER;
Definition~\ref{def:tar}).  Section~\ref{sec:framework} derives bounds on TMR
as a function of $\alpha_1, \ldots, \alpha_K$ and the inter-step
correlation structure.

\section{Trajectory-Level Conformal Framework}
\label{sec:framework}

Five proofs that are mechanical given their stated theorem or
proposition (Second-Order Degeneracy Threshold, Dependence Duality,
Finite-Sample Certifiability Gap, Sample Complexity of Beating the
Nominal Target, and, in Section~\ref{subsec:coupling_sources},
Coupling Provenance Decomposition) are given in full in the
supplementary material and summarized here; all theorem, proposition,
and corollary statements are complete in the main text.

We present the framework in seven parts.  Theorem~\ref{thm:bonf}
establishes the baseline Bonferroni bound as a special case of a
method-agnostic composition principle (Theorem~\ref{thm:general_composition})
that applies to any calibrated decision chain, not only conformal
prediction.  Theorem~\ref{thm:ie} gives the exact $K=2$
inclusion--exclusion identity.  Theorem~\ref{thm:hunter_trajectory}
gives a valid, certifiable upper bound for arbitrary $K$ that a naive
extension of the $K=2$ identity does not provide.
Section~\ref{subsec:certifiability} separates this structural bound from
its \emph{statistical certifiability} from a finite joint audit sample:
it distinguishes the oracle dependence gain a chain actually has from
the certifiable gain a given audit size can prove
(Theorem~\ref{thm:posthoc_certificate}), gives the convergence rate of
the resulting gap (Theorem~\ref{thm:certifiability_gap}), and gives
matching upper and information-theoretic lower bounds on the sample size
needed to certify any positive gain at all
(Theorems~\ref{thm:positive_gain_probability}--\ref{thm:positive_gain_lower_bound}).
Theorem~\ref{thm:pairwise_nonidentifiability} shows why pairwise
information alone cannot determine trajectory risk for $K\geq3$, and
Theorem~\ref{thm:sharp_k3_interval} gives the exact size of the
resulting identification gap.  Section~\ref{subsec:coupling_sources}
separates two mechanisms that can produce apparent inter-step coupling:
deterministic label nesting (Theorem~\ref{thm:label_induced_coupling})
and model-specific representation sharing versus shared sample
difficulty (Proposition~\ref{prop:coupling_decomposition}).
Theorem~\ref{thm:dependence_duality} explains, from first principles,
why positive inter-step dependence helps trajectory-level coverage even
though the same kind of dependence would hurt a redundant (parallel)
system.  Proposition~\ref{prop:equal} characterizes the optimal $\alpha$
allocation.

\subsection{Method-Agnostic Composition}
\label{subsec:mrcc}

\begin{definition}[Marginally Risk-Controlled Chain]
\label{def:mrcc}
A $K$-step decision chain is a collection of measurable failure events
$E_1, \ldots, E_K$ on a common probability space.  Let $Z_k = \mathbf{1}_{E_k}$
and $p_k = \Pr[E_k]$.  The chain is \emph{marginally risk-controlled}
at budgets $\alpha_1, \ldots, \alpha_K$ if $p_k \leq \alpha_k$ for all $k$.
No assumption is made on how the individual per-step rules are
constructed: split conformal prediction, conformal risk
control~\cite{angelopoulos2022conformal}, risk-controlling prediction
sets~\cite{bates2021distribution}, learn-then-test
calibration~\cite{angelopoulos2025learn}, and
non-conformal calibrated classifiers are all instances whenever they
provide a marginal guarantee of this form.
\end{definition}

\begin{theorem}[Universal and Conditional Composition]
\label{thm:general_composition}
For every marginally risk-controlled chain,
\begin{equation}
\label{eq:bonf}
\text{TMR} = \Pr\!\left[\bigcup_{k=1}^{K} E_k\right]
\leq \sum_{k=1}^{K} p_k
\leq \sum_{k=1}^{K} \alpha_k.
\end{equation}
If, in addition, the chain satisfies the stronger survival-conditional
guarantees
\begin{equation}
\label{eq:survival_conditional}
\Pr\!\left[E_k \,\middle|\, \bigcap_{j<k} E_j^c\right] \leq \alpha_k,
\qquad k = 1, \ldots, K,
\end{equation}
then
\begin{equation}
\label{eq:multiplicative_composition}
\text{TC} \geq \prod_{k=1}^{K} (1 - \alpha_k), \qquad
\text{TMR} \leq 1 - \prod_{k=1}^{K} (1 - \alpha_k).
\end{equation}
\end{theorem}

\begin{proof}
Inequality~\eqref{eq:bonf} is Boole's inequality (union bound) followed
by the marginal guarantees $p_k \leq \alpha_k$; it holds for arbitrary
events without independence or distributional assumptions.  For the
second statement, the chain rule of probability gives
\begin{equation}
\text{TC} = \Pr\!\left[\bigcap_{k=1}^{K} E_k^c\right]
= \prod_{k=1}^{K} \Pr\!\left[E_k^c \,\middle|\, \bigcap_{j<k} E_j^c\right]
\geq \prod_{k=1}^{K} (1 - \alpha_k),
\end{equation}
using~\eqref{eq:survival_conditional} on each factor.  Taking
complements proves~\eqref{eq:multiplicative_composition}.
\end{proof}

\begin{theorem}[Bonferroni Trajectory Bound]
\label{thm:bonf}
Setting $\alpha_k = \alpha/K$ in~\eqref{eq:bonf} ensures
$\text{TMR} \leq \alpha$.  This is the Bonferroni trajectory bound used
throughout this paper: it is distribution-free---valid for any joint
distribution of errors, with no estimated parameters.
\end{theorem}

While universally valid, Theorem~\ref{thm:bonf} treats per-step
errors as potentially adversarially combined.  When errors are
positively correlated---as we document in Section~\ref{sec:results}---the
bound is substantially loose.  The multiplicative
statement~\eqref{eq:multiplicative_composition} is strictly stronger but
requires the survival-conditional guarantee~\eqref{eq:survival_conditional},
which ordinary split CP's \emph{marginal} coverage does not automatically
provide; we return to this distinction in Proposition~\ref{prop:equal}.

\subsection{Correlation-Aware Inclusion--Exclusion Identity ($K=2$)}

Let $\rho_{kk'} = \text{Cor}(\mathbf{1}_{E_k}, \mathbf{1}_{E_{k'}})$
denote the Pearson correlation between the indicator variables of
miscoverage events at steps $k$ and $k'$.

\begin{theorem}[Inclusion--Exclusion Trajectory Identity]
\label{thm:ie}
For $K = 2$ steps with per-step miscoverage rates $p_k = \Pr[E_k]$ and
inter-step correlation $\rho_{12}\in[-1,1]$ (no sign restriction):
\begin{equation}
\label{eq:ie}
\text{TMR} = p_1 + p_2 - p_1 p_2
  - \rho_{12}\sqrt{p_1(1 - p_1)\, p_2(1 - p_2)}.
\end{equation}
Dropping the non-negative $-p_1p_2$ term from~\eqref{eq:ie} gives a
valid upper bound for any $\rho_{12}$:
\begin{equation}
\label{eq:ie_bound}
\text{TMR} \leq p_1 + p_2 - \rho_{12}\sqrt{p_1(1 - p_1)\, p_2(1 - p_2)}.
\end{equation}
This bound is only an \emph{improvement} over the Bonferroni sum
$p_1+p_2$ when $\rho_{12}\geq0$; for $\rho_{12}<0$ it is a valid but
looser statement than Bonferroni.  We report the plug-in
estimate~\eqref{eq:ie_alpha} uniformly across all configurations
without clipping to Bonferroni, so a handful of CIC-IDS-2018/DoS
configurations with $\hat\rho_{12}<0$ (Table~\ref{tab:correlation})
correctly show a plug-in value slightly above 0.10 in
Table~\ref{tab:bounds}---an operator should always report
$\min\{\text{Bonferroni}, \widehat{B}_{\mathrm{IE}}\}$ in practice, but
we leave the raw value visible here to show the theorem's actual
behavior under negative correlation.
In experiments, we report the following plug-in operational estimate:
\begin{equation}
\label{eq:ie_alpha}
\widehat{B}_{\mathrm{IE}} =
\alpha_1 + \alpha_2
  - \hat{\rho}_{12}\sqrt{\alpha_1(1-\alpha_1)\,\alpha_2(1-\alpha_2)},
\end{equation}
where $\hat{\rho}_{12}$ is estimated from held-out evaluation splits
after calibration.
\end{theorem}

\begin{proof}
By inclusion--exclusion:
\begin{equation}
\Pr[E_1 \cup E_2] = \Pr[E_1] + \Pr[E_2] - \Pr[E_1 \cap E_2].
\end{equation}
The joint probability can be written via the Pearson correlation of
Bernoulli indicators:
\begin{equation}
\rho_{12} = \frac{\Pr[E_1 \cap E_2] - p_1 p_2}{\sqrt{p_1(1-p_1)\, p_2(1-p_2)}},
\end{equation}
so $\Pr[E_1 \cap E_2] = p_1 p_2 + \rho_{12}\sqrt{p_1(1-p_1)\,p_2(1-p_2)}$.
Substituting yields~\eqref{eq:ie}.  Inequality~\eqref{eq:ie_bound}
drops the $-p_1 p_2$ term (non-positive, since $p_1p_2\geq0$) to obtain
a simpler, larger---hence still valid as an upper bound---expression.  Equation~\eqref{eq:ie_alpha} is a plug-in version used for
empirical tightness analysis; its validity depends on the accuracy and
conservativeness of $\hat{\rho}_{12}$.
\end{proof}

\begin{remark}[Eq.~\eqref{eq:ie_alpha} is an estimate, not a bound]
\label{rem:tight}
Unlike Theorem~\ref{thm:bonf}, the plug-in quantity
$\widehat{B}_{\mathrm{IE}}$ in~\eqref{eq:ie_alpha} relies on the
calibration estimate $\hat{\rho}_{12}$ and is therefore an
\emph{empirically calibrated point estimate}, not a distribution-free
or certified guarantee: when $\hat{\rho}_{12}$ overestimates the true
population correlation, the empirical TMR can exceed
$\widehat{B}_{\mathrm{IE}}$, which we report as a residual (not a
``bound violation'') throughout the paper.  The estimate is close to
the true TMR when (i)~the per-step CP miscoverage rate is close to
$\alpha_k$ (saturated CP) and (ii)~$\hat{\rho}_{12}$ accurately
estimates the population correlation.  Section~\ref{subsec:hunter_empirical}
gives a certified alternative that does not require either condition.
\end{remark}

\subsection{Why the Naive $K>2$ Extension Fails, and a Valid Alternative}
\label{subsec:hunter}

A natural guess for $K>2$ is to subtract every pairwise joint-failure
term from the Bonferroni sum.  Writing $\hat{q}_{kk'} \triangleq
\alpha_k\alpha_{k'} + \hat{\rho}_{kk'}\sqrt{\alpha_k(1-\alpha_k)\,\alpha_{k'}(1-\alpha_{k'})}$
for the plug-in estimate of the pairwise overlap $q_{kk'}$ (the same
identity used in Theorem~\ref{thm:ie}), this gives
\begin{equation}
\label{eq:ie_general}
S_1 - S_2 \triangleq \sum_{k=1}^{K} \alpha_k
  - \sum_{k < k'} \hat{q}_{kk'}.
\end{equation}
This expression is \emph{not} a valid upper bound on TMR for $K>2$; the
correct classical statement (Bonferroni's inequalities on the alternating
inclusion--exclusion partial sums) is that
$S_1 - S_2$ is a \emph{lower} bound on TMR, not an upper bound:

\begin{theorem}[Bonferroni Parity]
\label{thm:bonferroni_parity}
Let $S_r = \sum_{i_1 < \cdots < i_r} \Pr[E_{i_1} \cap \cdots \cap E_{i_r}]$.
For every $m$ for which the sums are defined,
\begin{equation}
\sum_{r=1}^{2m} (-1)^{r+1} S_r \;\leq\; \text{TMR} \;\leq\;
\sum_{r=1}^{2m-1} (-1)^{r+1} S_r.
\end{equation}
In particular $S_1 - S_2$ is a lower bound, not an upper bound,
whenever $K > 2$.
\end{theorem}

The clipped quantity $L_2 = \max\{0, S_1 - S_2\}$ can hit zero well
before TMR is actually small.  A minimal example makes this concrete:
if $E_1 = \cdots = E_K$ (all steps fail on exactly the same trajectories),
the true union is $\Pr[E_1] = p$ for every $K$, but $S_1 - S_2 = Kp -
\binom{K}{2}p$ becomes non-positive starting at $K=3$.  The following
result characterizes exactly when this clipped quantity degenerates.

\begin{theorem}[Second-Order Degeneracy Threshold]
\label{thm:pairwise_degeneracy}
Under the homogeneous model $p_i = \alpha/K$ and $\rho_{ij} = \rho \geq 0$
for all $i \neq j$, $L_2 = 0$ if and only if
\begin{equation}
\label{eq:homogeneous_threshold}
(K-1)\left[\rho + (1-\rho)\frac{\alpha}{K}\right] \geq 2.
\end{equation}
The smallest degenerate integer depth is the smallest integer $K$
satisfying $K \geq K_+(\alpha,\rho)$, where, writing
$D(\alpha,\rho) \triangleq \rho + 2 - (1-\rho)\alpha$ for brevity,
\begin{equation}
K_+(\alpha,\rho) = \frac{D(\alpha,\rho) +
\sqrt{D(\alpha,\rho)^2 + 4\rho(1-\rho)\alpha}}{2\rho}.
\end{equation}
As $\alpha \downarrow 0$, $K_+(\alpha,\rho) = 1 + 2/\rho + O(\alpha)$.
At $\alpha = 0.10$, $\rho = 0.298$, $K_+ \approx 7.51$, so $K=3$ through
$K=7$ all remain positive and $K=8$ is the first degenerate
depth---consistent with the value in Table~\ref{tab:scaling}.
($\bar\rho=0.298$ is the mean over the 12 CIC-IDS-2018/DoS and
RT-IoT2022/Probe configurations at $\alpha=0.10$ used throughout
Section~\ref{subsec:bounds} and Table~\ref{tab:correlation}.)
\end{theorem}

\begin{proof}
Algebraic; solving $L_2=0$ for the homogeneous model yields the stated
threshold, and treating it as an equality in $K$ gives a quadratic
whose positive root is $K_+(\alpha,\rho)$ (full derivation in the
supplementary material).
\end{proof}

\begin{remark}[This threshold is a truncation artifact, not a structural limit]
\label{rem:truncation_artifact}
$K_+$ characterizes when the \emph{even-order truncation} $S_1 - S_2$
stops being informative, not when trajectory risk itself becomes
unrecoverable.  Full inclusion--exclusion is exact for every finite $K$;
the limitation is that pairwise correlation information alone stops
determining TMR once $K \geq 3$, which we make precise next.
\end{remark}

A valid alternative uses joint pairwise \emph{overlap} probabilities
$q_{ij} = \Pr[E_i \cap E_j]$ directly, together with a spanning tree over
the $K$ steps~\cite{hunter1976upper}:

\begin{theorem}[Spanning-Tree Pairwise Upper Bound]
\label{thm:hunter_trajectory}
Let $T$ be any spanning tree on vertices $\{1,\ldots,K\}$.  Then
\begin{equation}
\label{eq:hunter_trajectory}
\text{TMR} \leq \sum_{k=1}^{K} p_k - \sum_{(i,j) \in T} q_{ij}.
\end{equation}
Consequently, if $p_k \leq \alpha_k$ and simultaneous lower confidence
bounds $q_{ij} \geq \ell_{ij}$ are available from held-out data, then
with the corresponding confidence,
\begin{equation}
\label{eq:hunter_certified}
\text{TMR} \leq \min\!\left\{1,\; \sum_{k=1}^{K} \alpha_k -
\max_{T \in \mathcal{T}_K} \sum_{(i,j)\in T} \ell_{ij}\right\},
\end{equation}
where the maximizing $T$ is a maximum-weight spanning tree under edge
weights $\ell_{ij}$.  This bound uses only $K-1$ certified pairwise
overlaps and remains valid for every $K$, with no third-order terms.
\end{theorem}

\begin{proof}
Root $T$ at an arbitrary vertex and order every parent before its
children; let $\pi(k)$ denote the parent of non-root vertex $k$.  The
exact disjoint-increment decomposition of a union over this order gives
$\Pr[\bigcup_k E_k] = p_{\text{root}} + \sum_{k \neq \text{root}}
\Pr[E_k \setminus \bigcup_{j<k} E_j] \leq p_{\text{root}} +
\sum_{k \neq \text{root}} \Pr[E_k \setminus E_{\pi(k)}] =
\sum_k p_k - \sum_{k \neq \text{root}} q_{k,\pi(k)}$, where the
inequality holds because $\bigcup_{j<k} E_j \supseteq E_{\pi(k)}$.
This proves~\eqref{eq:hunter_trajectory}; substituting the certified
bounds and optimizing over $T$ proves~\eqref{eq:hunter_certified}.
\end{proof}

Unlike the Pearson-correlation plug-in~\eqref{eq:ie_alpha}, $q_{ij}$
remains well-defined even when $p_i$ is near zero and admits an exact
one-sided confidence bound (e.g.\ Clopper--Pearson) with no normality
approximation.  Section~\ref{subsec:hunter_empirical} computes $q_{12}$
and its lower confidence bound for our $K=2$ pipeline.

\subsection{Realized-Chain Risk Certification and the Certifiability Gap}
\label{subsec:certifiability}

Theorem~\ref{thm:hunter_trajectory} and~\eqref{eq:hunter_certified}
already sketch a certified version of the spanning-tree bound.  We now
make the two-level structure of that certificate explicit, because our
own experiments (Section~\ref{subsec:hunter_empirical}) conflated two
different objects: a \emph{marginal} guarantee averaged over calibration
randomness, and a \emph{realized-chain} guarantee for the one specific
trained-and-calibrated pipeline that is actually deployed.

Let $\Theta$ denote all randomness used to train and calibrate the
deployed chain, and write $p_k(\theta)$, $q_{ij}(\theta)$, and
$R(\theta) = \Pr_\theta[\bigcup_k E_k]$ for the corresponding
quantities conditional on a realized chain $\Theta=\theta$.

\begin{definition}[Marginal and realized-chain trajectory risk]
\label{def:two_level_risk}
The marginal trajectory risk averages over both calibration randomness
and a fresh trajectory, $R_{\mathrm{marg}} = \mathbb{E}_\Theta[R(\Theta)]$.
The realized-chain risk $R(\theta)$ is the deployment risk of the
particular trained and calibrated chain that is actually deployed.
\end{definition}

\begin{proposition}[What marginal stage-wise guarantees compose]
\label{prop:marginal_composition}
If $\mathbb{E}_\Theta[p_k(\Theta)] \leq \alpha_k$ for $k=1,\ldots,K$,
then $R_{\mathrm{marg}} \leq \sum_k \alpha_k$ (Theorem~\ref{thm:bonf},
restated at the marginal level).  This conclusion concerns
$R_{\mathrm{marg}}$ and does \emph{not} imply $R(\theta) \leq \sum_k
\alpha_k$ for every realized chain $\theta$.
\end{proposition}

The remainder of this subsection conditions on a fixed realized chain
$\theta$ and asks when a finite joint audit sample can certify
$R(\theta)$ below the marginal-only bound---the relevant question when
deciding whether one particular deployed chain is safe enough to
automate.  Let $\mathcal{A}_n = \{\bm{Z}^{(t)}\}_{t=1}^n$ be an i.i.d.\
audit sample from the fixed chain, independent of the data used to
select any threshold, and write $\widehat{p}_k = \frac{1}{n}\sum_t
Z_k^{(t)}$, $\widehat{q}_{ij} = \frac{1}{n}\sum_t Z_i^{(t)}Z_j^{(t)}$.
The framework accepts any construction satisfying
\begin{align}
\Pr_{\mathcal{A}_n}[p_k \leq U_k \ \forall k] &\geq 1-\delta_p,
\label{eq:sim_marg}\\
\Pr_{\mathcal{A}_n}[L_{ij} \leq q_{ij} \ \forall i<j] &\geq 1-\delta_q,
\label{eq:sim_overlap}
\end{align}
with no independence required between the two families.

\begin{definition}[Oracle and certifiable dependence gain]
\label{def:certifiable_gain}
With $\mathcal{T}_K$ the spanning trees on $\{1,\ldots,K\}$, the oracle
pairwise dependence gain is $G_{\mathrm{oracle}} = \max_{T\in\mathcal{T}_K}
\sum_{(i,j)\in T} q_{ij}$, and the finite-sample certifiable dependence
gain is $G_{\mathrm{cert}} = \max_{T\in\mathcal{T}_K} \sum_{(i,j)\in T}
L_{ij}$.  The dependence certifiability gap is $\Delta_{\mathrm{dep}} =
G_{\mathrm{oracle}} - G_{\mathrm{cert}}$.
\end{definition}

\begin{theorem}[Post-hoc modular trajectory certificate]
\label{thm:posthoc_certificate}
Assume~\eqref{eq:sim_marg} and~\eqref{eq:sim_overlap}, and let $\widehat{T}
\in \arg\max_{T\in\mathcal{T}_K} \sum_{(i,j)\in T} L_{ij}$ be chosen from
the audit data.  Then, with probability at least $1-\delta_p-\delta_q$,
\begin{equation}
\label{eq:posthoc_certificate}
R(\theta) \leq B_{\mathrm{cert}} \triangleq \min\!\left\{1,\;
\sum_{k=1}^K U_k - G_{\mathrm{cert}}\right\}.
\end{equation}
The data-dependent choice of $\widehat{T}$ does not invalidate the
certificate, because Theorem~\ref{thm:hunter_trajectory}'s spanning-tree
bound holds \emph{simultaneously for every} $T \in \mathcal{T}_K$, and
\eqref{eq:sim_marg}--\eqref{eq:sim_overlap} sandwich $p_k, q_{ij}$
simultaneously over all $k$ and all pairs; selecting the maximizing tree
after seeing the data therefore stays inside the same validity event.
\end{theorem}

\begin{proof}
For any fixed $T$, Theorem~\ref{thm:hunter_trajectory} gives $R(\theta)
\leq \sum_k p_k - \sum_{(i,j)\in T} q_{ij}$.  On the event where
\eqref{eq:sim_marg}--\eqref{eq:sim_overlap} both hold, $\sum_k p_k -
\sum_{(i,j)\in T} q_{ij} \leq \sum_k U_k - \sum_{(i,j)\in T} L_{ij}$
simultaneously for every $T$, hence in particular for $\widehat{T}$.  A
union bound over the two validity events gives probability at least
$1-\delta_p-\delta_q$; clipping at one preserves validity.
\end{proof}

\begin{corollary}[Exact two-stage form]
\label{cor:k2_certificate}
For $K=2$, $R(\theta) = p_1+p_2-q_{12}$ exactly, and with probability at
least $1-\delta_p-\delta_q$, $R(\theta) \leq U_1+U_2-L_{12}$.  The
certifiable dependence gain over the marginal-only certificate
$U_1+U_2$ is exactly $L_{12}$.
\end{corollary}

\begin{remark}[A direct audit of the union event is a distinct, tighter
alternative at $K=2$]
\label{rem:direct_union_audit}
Corollary~\ref{cor:k2_certificate} decomposes the certificate through
three separately-audited quantities ($U_1$, $U_2$, $L_{12}$), which is
necessary when only pairwise summaries are available or when $K>2$
requires selecting among several candidate trees.  When $K=2$ and the
full joint failure vector $(Z_1,Z_2)$ is observed on every audit
trajectory---exactly our setting---one may instead construct a single
exact one-sided Clopper--Pearson upper bound $U_\cup$ directly on
$r=\Pr[E_1\cup E_2\mid\mathcal{C}]$ from the Binomial count
$X_\cup=\sum_t \mathbf{1}\{Z_1^{(t)}\cup Z_2^{(t)}\}$, at the same
overall confidence level.  The decomposed certificate pays the CP
slack of \emph{three} separately-constructed intervals (two upper, one
lower) combined by a union bound over their confidence budgets, while
$U_\cup$ pays the slack of \emph{one} interval on the compound event
directly; combining marginal bounds this way is generally weakly more
conservative than a direct joint construction at matched confidence,
and $U_\cup < U_1+U_2-L_{12}$ holds in every one of our 12
configurations (Section~\ref{subsec:hunter_empirical}).  Both are valid
$1-\delta$ certificates for $R(\theta)$; $U_\cup$ is the one an
operator with full joint audit data should actually deploy, and it is
the quantity we report as the primary result in
Section~\ref{subsec:hunter_empirical}.  We report the decomposed
certificate alongside it because it is the object
Theorem~\ref{thm:posthoc_certificate} formally proves and the one that
generalizes to $K>2$ and to pairwise-only audit data; the gap between
$U_\cup$ and $U_1+U_2-L_{12}$ is itself informative about how much
conservatism the decomposition costs.
\end{remark}

\begin{remark}[Overlap, not correlation, is the certification primitive]
As in~\eqref{eq:ie}, $q_{ij} = p_ip_j + \rho_{ij}\sqrt{p_i(1-p_i)p_j(1-p_j)}$
relates overlap to correlation, but $q_{ij}$ is the quantity that
directly enters~\eqref{eq:posthoc_certificate}.  It stays well defined
when a marginal rate is near zero and admits an exact one-sided
binomial confidence bound with no normality approximation---Pearson
correlation should be reported descriptively, not used as the
certification object itself.
\end{remark}

Let $M = \binom{K}{2}$.  For $\delta_p, \delta_q \in (0,1)$, Hoeffding's
inequality with a union bound gives valid simultaneous bounds
\begin{align}
U_k &= \min\{1, \widehat{p}_k + a_p\}, & a_p &= \sqrt{\tfrac{\log(K/\delta_p)}{2n}},
\label{eq:hoeff_up}\\
L_{ij} &= \max\{0, \widehat{q}_{ij} - a_q\}, & a_q &= \sqrt{\tfrac{\log(M/\delta_q)}{2n}}.
\label{eq:hoeff_lo}
\end{align}

\begin{theorem}[Closed-form finite-sample validity]
\label{thm:closed_form_validity}
The quantities in \eqref{eq:hoeff_up}--\eqref{eq:hoeff_lo} satisfy
\eqref{eq:sim_marg}--\eqref{eq:sim_overlap}, so
\eqref{eq:posthoc_certificate} is a valid $1-\delta_p-\delta_q$
certificate for $R(\theta)$, with no independence required among stages
or among edge statistics.
\end{theorem}

\begin{proof}
For each $k$, Hoeffding's inequality gives $\Pr[\widehat{p}_k-p_k<-a_p]
\leq e^{-2na_p^2}=\delta_p/K$; a union bound over $K$ stages gives
simultaneous validity of all $U_k$.  Likewise $\Pr[\widehat{q}_{ij}-q_{ij}>a_q]
\leq e^{-2na_q^2}=\delta_q/M$ for each pair, and a union bound over the
$M$ pairs gives simultaneous validity of all $L_{ij}$.
\end{proof}

Writing $b_p, b_q$ for the same expressions with $\delta_p,\delta_q$
replaced by second confidence levels $\beta_p,\beta_q$:

Let $B_{\mathrm{oracle}}=\min\{1,\sum_k p_k-G_{\mathrm{oracle}}\}$ denote
the certificate an infinite audit sample would recover.

\begin{theorem}[Finite-sample certifiability gap]
\label{thm:certifiability_gap}
With probability at least $1-\delta_q-\beta_q$,
\begin{equation}
\label{eq:dep_gap_bound}
0 \leq \Delta_{\mathrm{dep}} \leq (K-1)(a_q+b_q)
\end{equation}
before clipping at one.  Consequently, a sufficient audit size for
$\Delta_{\mathrm{dep}} \leq \eta$ is
\begin{equation}
\label{eq:eta_complexity}
n \geq \frac{(K-1)^2}{2\eta^2}\left[\sqrt{\log(M/\delta_q)} +
\sqrt{\log(M/\beta_q)}\right]^2.
\end{equation}
Separately, with probability at least
$1-\delta_p-\delta_q-\beta_p-\beta_q$ (a union bound over all four
validity events),
\begin{equation}
\label{eq:total_gap}
0 \leq B_{\mathrm{cert}}-B_{\mathrm{oracle}}
\leq K(a_p+b_p)+(K-1)(a_q+b_q),
\end{equation}
before clipping at one---the marginal-side terms $\delta_p,\beta_p$
enter only this combined statement, not the $\Delta_{\mathrm{dep}}$-only
bound~\eqref{eq:dep_gap_bound}.
\end{theorem}

\begin{proof}
A second Hoeffding application at confidence levels $\beta_q$ (resp.
$\beta_p$) bounds $\widehat q_{ij}$ (resp. $\widehat p_k$) away from
$q_{ij}$ (resp. $p_k$) with high probability; combined with
$L_{ij}\geq q_{ij}-(a_q+b_q)$ summed over the oracle-maximizing tree's
edges, this gives both bounds by a union bound over the relevant
validity events (full derivation in the supplementary material).
\end{proof}

\paragraph{When can any positive gain be certified at all?}
The convergence rate above controls the \emph{magnitude} of the
certifiability gap once some overlap is detected; a separate question is
how much audit data is needed to detect any overlap in the first place,
which is what our own $K=2$ pipeline runs into directly.  For $K=2$, let
$X = \sum_t Z_1^{(t)}Z_2^{(t)} \sim \mathrm{Binomial}(n,q_{12})$ and let
$L_{\mathrm{CP}}$ be the exact one-sided Clopper--Pearson lower bound for
$q_{12}$.

\begin{theorem}[Exact positive-gain probability]
\label{thm:positive_gain_probability}
$L_{\mathrm{CP}}(X;n,\delta) > 0 \iff X \geq 1$, hence
$\Pr_{q_{12}}[G_{\mathrm{cert}} > 0] = 1-(1-q_{12})^n$.  Obtaining a
strictly tighter certificate than the marginal-only bound with
probability at least $1-\beta$ requires
$n \geq \log\beta / \log(1-q_{12})$, i.e.\ $n \approx \log(1/\beta)/q_{12}$
when $q_{12}$ is small.
\end{theorem}

\begin{proof}
The exact one-sided Clopper--Pearson bound equals zero exactly when
$X=0$ and is a positive beta quantile otherwise, giving the
equivalence.  $\Pr[X=0] = (1-q_{12})^n$ gives the stated probability;
solving $1-(1-q_{12})^n \geq 1-\beta$ for $n$, using
$\log(1-q)=-q+o(q)$ for the small-$q$ approximation, completes the proof.
\end{proof}

This is not an artifact of the Clopper--Pearson construction: it is a
rate that any valid procedure must pay.

\begin{theorem}[Information-theoretic limit on detecting positive overlap]
\label{thm:positive_gain_lower_bound}
Let $L\in[0,1]$ be any procedure based on $n$ i.i.d.\ Bernoulli overlap
observations satisfying $\inf_{q}\Pr_q[L\leq q]\geq 1-\delta$, with
$\beta+\delta<1$.  If $\Pr_q[L>0]\geq 1-\beta$ at some $q>0$, then
$n \geq \log(\beta+\delta)/\log(1-q)$.
\end{theorem}

\begin{proof}
Validity at $q=0$ gives $\Pr_0[L>0]\leq\delta$.  Since $P_0^n$
concentrates on the all-zero sample, $\mathrm{TV}(P_0^n,P_q^n) =
1-(1-q)^n$, so $\Pr_q[L>0] \leq \delta + 1-(1-q)^n$.  Combining with
$\Pr_q[L>0]\geq 1-\beta$ gives $(1-q)^n \leq \beta+\delta$, which
rearranges to the stated bound.
\end{proof}

Theorem~\ref{thm:positive_gain_probability} is therefore rate-optimal in
the rare-overlap regime: certifying \emph{any} dependence gain costs
$\Theta(1/q_{12})$ audit trajectories, while certifying a gain to
additive precision $\eta$ costs the usual $\Theta(1/\eta^2)$ via a
Bretagnolle--Huber testing argument applied to two overlap levels
$q_0<q_1$, which gives $n \geq \log(1/[2(\delta+\beta)]) /
\mathrm{kl}(q_0\|q_1) = \Omega\!\left(\tfrac{q_0(1-q_0)}{\eta^2}
\log\tfrac{1}{\delta+\beta}\right)$ for $q_1=q_0+2\eta$ and small $\eta$.
Section~\ref{subsec:hunter_empirical} reports where our two IDS
attack-variant tasks fall relative to this threshold.

\paragraph{Beating the nominal target is a different, harder event than
detecting a positive gain.}  Theorems~\ref{thm:positive_gain_probability}--\ref{thm:positive_gain_lower_bound}
characterize when a dependence-aware certificate can be shown to be
\emph{strictly tighter than the marginal-only certificate it replaces}.
A practitioner instead usually wants to know when the certificate beats
the nominal budget $\alpha$ itself.  For $K=2$, writing
$R=p_1+p_2-q_{12}$ for the true realized-chain risk
(Corollary~\ref{cor:k2_certificate}) and $\gamma=\alpha-R>0$ for the
true safety margin, and letting $\widehat{R}=\frac1n\sum_t Z_\cup^{(t)}$
be the direct empirical estimate of the union-event rate on $n$ i.i.d.\
audit trajectories with Clopper--Pearson (equivalently Hoeffding) upper
bound $U_\cup$:

\begin{theorem}[Sample complexity of beating the nominal target]
\label{thm:nominal_crossing_complexity}
A sufficient condition for $\Pr[U_\cup<\alpha]\geq1-\beta$, using the
Hoeffding-form bound $U_\cup=\widehat R+\sqrt{\log(1/\delta)/(2n)}$ at
confidence $1-\delta$, is
\begin{equation}
\label{eq:gamma_upper}
n \geq \frac{\left(\sqrt{\log(1/\delta)}+\sqrt{\log(1/\beta)}\right)^2}{2\gamma^2}.
\end{equation}
Conversely, for any procedure $U$ satisfying $\inf_{r}\Pr_r[U\geq r]\geq1-\delta$
(uniform validity as an upper bound at every true rate $r$), if
$\Pr_{R}[U<\alpha]\geq1-\beta$ at true rate $R=\alpha-\gamma$, then
\begin{equation}
\label{eq:gamma_lower}
n \geq \frac{\log\!\left(1/[2(\delta+\beta)]\right)}{\mathrm{kl}(\alpha\,\|\,\alpha-\gamma)}
= \Omega\!\left(\frac{\alpha(1-\alpha)}{\gamma^2}\log\frac{1}{\delta+\beta}\right)
\end{equation}
for small $\gamma$.  Both~\eqref{eq:gamma_upper} and~\eqref{eq:gamma_lower}
scale as $\Theta(1/\gamma^2)$, matching up to the confidence-dependent
constant---the usual estimation rate, not the $\Theta(1/q_{12})$
detection rate of Theorem~\ref{thm:positive_gain_probability}.
\end{theorem}

\begin{proof}
The upper bound follows from Hoeffding's inequality applied to
$\widehat R$; the lower bound follows from the Bretagnolle--Huber
testing inequality applied to the Bernoulli union event $Z_\cup$,
comparing true rates $\alpha$ and $\alpha-\gamma$ (full derivation in
the supplementary material).
\end{proof}

\begin{corollary}[Positive gain and nominal crossing are distinct events]
\label{cor:nominal_crossing}
$L_{12}>0$ (Theorem~\ref{thm:positive_gain_probability}) is strictly
weaker than $U_1+U_2-L_{12}<\alpha$ whenever the marginal-only
certificate $U_1+U_2$ itself already exceeds $\alpha$: the two events
have different sample-complexity rates,
$\Theta(1/q_{12})$ versus $\Theta(1/\gamma^2)$, so a configuration can
certify a strictly positive dependence gain
(Theorem~\ref{thm:positive_gain_probability}) without its final
certificate beating Bonferroni (Theorem~\ref{thm:nominal_crossing_complexity})---this
is not an inconsistency between the two claims, since they are
different statistical events with different evidence requirements.
\end{corollary}

Why can pairwise information not be pushed further?  Because, in
general, it cannot---even with exact pairwise probabilities, not just
correlations.

\begin{theorem}[Pairwise Non-Identifiability for $K \geq 3$]
\label{thm:pairwise_nonidentifiability}
Marginal probabilities and all pairwise joint probabilities do not
determine TMR once $K \geq 3$.  For $K=3$, there exist two exchangeable
laws with identical $p_1=p_2=p_3=0.5$, identical $q_{12}=q_{13}=q_{23}=0.27$
(hence identical $\rho_{ij}=0.08$ for every pair), but with trajectory
failure probabilities $0.73$ and $0.96$, respectively---straddling the
independence value $1-(1-0.5)^3 = 0.875$ on both sides.
\end{theorem}

\begin{proof}
Let $N = Z_1+Z_2+Z_3$ and, conditional on $N=n$, distribute probability
uniformly over the $\binom{3}{n}$ binary vectors with exactly $n$ ones
(this preserves exchangeability). Under law A: $\Pr[N{=}0]=0.27$,
$\Pr[N{=}2]=0.69$, $\Pr[N{=}3]=0.04$.  Under law B: $\Pr[N{=}0]=0.04$,
$\Pr[N{=}1]=0.69$, $\Pr[N{=}3]=0.27$.  Both give
$p_i = \mathbb{E}[N]/3 = 0.5$ and $q_{ij} = \mathbb{E}\binom{N}{2}/3=0.27$,
hence identical marginals and pairwise correlations.  Yet
$\Pr[N>0] = 1-\Pr[N{=}0]$ equals $0.73$ under A and $0.96$ under B.
\end{proof}

\begin{remark}
This is precisely why we report the exact $K=2$ identity
(Theorem~\ref{thm:ie}) as empirically validated, use the
distribution-free Hunter bound (Theorem~\ref{thm:hunter_trajectory}) for
any $K$ we do measure, and do not attempt to certify TMR at $K \geq 3$
from pairwise correlations alone---Table~\ref{tab:scaling}'s $K=3$ row
is reported strictly as an illustrative extrapolation, consistent with
Remark~\ref{rem:truncation_artifact}.
\end{remark}

The counterexample above shows non-identifiability qualitatively; for
$K=3$ the exact size of the resulting uncertainty has a closed form.

\begin{theorem}[Sharp identified interval for $K=3$]
\label{thm:sharp_k3_interval}
Let $S_1=p_1+p_2+p_3$, $S_2=q_{12}+q_{13}+q_{23}$, and
$t=\Pr[E_1\cap E_2\cap E_3]$, so $R=S_1-S_2+t$ by inclusion--exclusion.
Given feasible $p_i$ and $q_{ij}$, define
\begin{align}
t_{\min} &= \max\{0,\; q_{12}{+}q_{13}{-}p_1,\; q_{12}{+}q_{23}{-}p_2,\;
q_{13}{+}q_{23}{-}p_3\},\\
t_{\max} &= \min\{q_{12},q_{13},q_{23},\; 1-S_1+S_2\}.
\end{align}
Then $R \in [S_1-S_2+t_{\min},\, S_1-S_2+t_{\max}]$, and every value in
this interval is attained by some joint distribution with the specified
marginals and pairwise overlaps.
\end{theorem}

\begin{proof}
The eight atom probabilities of $(Z_1,Z_2,Z_3)$ can be written in terms
of $t$ as $\pi_{111}=t$; $\pi_{110}=q_{12}-t$, $\pi_{101}=q_{13}-t$,
$\pi_{011}=q_{23}-t$; $\pi_{100}=p_1-q_{12}-q_{13}+t$ (cyclically for
$\pi_{010},\pi_{001}$); and $\pi_{000}=1-S_1+S_2-t$.  Nonnegativity of
all eight atoms is equivalent to $t_{\min}\leq t\leq t_{\max}$, and
conversely every $t$ in that range yields a valid joint law with the
required marginals and pairwise overlaps.  Since $R=1-\pi_{000}=S_1-S_2+t$,
the interval is both valid and sharp.
\end{proof}

\begin{corollary}[Exact non-identifiability condition]
\label{cor:k3_nonidentifiability}
For $K=3$, marginals and pairwise overlaps identify $R$ exactly iff
$t_{\min}=t_{\max}$; otherwise the irreducible pairwise-information width
is $W_{\mathrm{pair}} = t_{\max}-t_{\min} > 0$.  The upper endpoint of the
sharp interval coincides with the clipped spanning-tree bound
(Theorem~\ref{thm:hunter_trajectory}) at $K=3$, so that bound is sharp
among all bounds using only marginals and pairwise overlaps; the
interval \emph{width} is what quantifies the information lost by not
observing the triple overlap.  For general $K$, the analogous sharp
interval is the solution to a linear program over the $2^K$ joint atom
probabilities, subject to the marginal and pairwise-overlap constraints.
\end{corollary}

The results of this section, together with
Section~\ref{subsec:certifiability}, separate three distinct quantities
that are easy to conflate under a single label of ``dependence gain.''
Let $G_{\mathrm{true}} \triangleq \sum_k p_k - R$ be the actual
reduction a realized chain's true dependence structure achieves over
the marginal sum, where $R=\Pr[\bigcup_k E_k]$ is the chain's true
trajectory risk; recall $G_{\mathrm{oracle}}$ (pairwise-only, exact
overlaps) and $G_{\mathrm{cert}}$ (pairwise-only, finite-sample lower
bounds) from Definition~\ref{def:certifiable_gain}.

\begin{theorem}[Structural and statistical gain decomposition]
\label{thm:gain_decomposition}
$0 \leq G_{\mathrm{cert}} \leq G_{\mathrm{oracle}} \leq G_{\mathrm{true}}$,
and
\begin{equation}
\label{eq:gain_decomposition}
G_{\mathrm{true}} - G_{\mathrm{cert}}
= \underbrace{(G_{\mathrm{true}} - G_{\mathrm{oracle}})}_{\text{structural gap}}
+ \underbrace{(G_{\mathrm{oracle}} - G_{\mathrm{cert}})}_{\text{certifiability gap}}.
\end{equation}
For $K=2$, $G_{\mathrm{true}}=G_{\mathrm{oracle}}=q_{12}$ exactly (the
structural gap vanishes: there is only one pair, so pairwise
information is already everything there is to know), and the entire
decomposition reduces to the certifiability gap of
Theorem~\ref{thm:certifiability_gap}.  For $K\geq3$, the structural gap
is generically strictly positive and does not vanish as $n\to\infty$:
by Corollary~\ref{cor:k3_nonidentifiability}, $G_{\mathrm{true}}$ is
not even identified by exact marginals and pairwise overlaps, so no
amount of pairwise-only audit data---however large---can close it.
\end{theorem}

\begin{proof}
$G_{\mathrm{cert}}\leq G_{\mathrm{oracle}}$ holds because $L_{ij}\leq
q_{ij}$ for every edge under the validity event
(Theorem~\ref{thm:certifiability_gap}'s proof).  $G_{\mathrm{oracle}}\leq
G_{\mathrm{true}}$ is the spanning-tree bound of
Theorem~\ref{thm:hunter_trajectory} restated as a gain: $R\leq\sum_k
p_k-G_{\mathrm{oracle}}$ rearranges to $G_{\mathrm{oracle}}\leq\sum_k
p_k-R=G_{\mathrm{true}}$.  Equation~\eqref{eq:gain_decomposition} is
immediate algebra.  For $K=2$ there is a unique pair and $R=p_1+p_2-q_{12}$
exactly (Corollary~\ref{cor:k2_certificate}), so $G_{\mathrm{true}}=q_{12}=G_{\mathrm{oracle}}$.
For $K\geq3$, Corollary~\ref{cor:k3_nonidentifiability} exhibits, for
fixed marginals and pairwise overlaps, joint laws with different $R$
(hence different $G_{\mathrm{true}}$) and identical $G_{\mathrm{oracle}}$;
since no audit of pairwise statistics alone can distinguish these laws,
the structural gap persists regardless of audit size.
\end{proof}

This decomposition is why Sections~\ref{subsec:certifiability}
and~\ref{subsec:coupling_sources}--\ref{subsec:hunter_empirical} are one
question, not two: the certifiability gap asks how much of the
dependence a pairwise audit \emph{could} in principle reveal is lost to
finite-sample noise, while the structural gap asks how much of the
chain's true dependence a pairwise audit could never reveal even with
infinite data, because it lives in interactions the audit does not
observe. Every dependence-aware bound in this paper pays one, the
other, or both costs, and Table~\ref{tab:mechanism}'s finding that
same-model pairing adds no material covariance beyond cross-model
pairing (Section~\ref{subsec:hunter_empirical}) says the residual
coupling we do certify is fully explained by pairwise, shared-difficulty
structure---for this $K=2$ pipeline, there is no hidden higher-order
term for the structural gap to hide in.

\subsection{Two Sources of Apparent Coupling}
\label{subsec:coupling_sources}

Two mechanisms can produce the strongly positive $\hat\rho$ we
originally measured (Section~\ref{subsec:hunter_empirical}), and it
matters which one is at work: one is a property of task construction,
the other a property of the model.

\paragraph{Deterministic label nesting.}
Let $Y_f$ be a fine-grained label and $Y_c=g(Y_f)$ its coarse label, with
prediction sets $C_f(x), C_c(x)$ and failure events $E_f=\{Y_f\notin
C_f(X)\}$, $E_c=\{Y_c\notin C_c(X)\}$.

\begin{theorem}[Deterministic label nesting induces nested failures]
\label{thm:label_induced_coupling}
If the prediction sets are label-consistent, $g(C_f(x))\subseteq C_c(x)$
for every $x$, then $E_c\subseteq E_f$ and $q_{cf}=\Pr[E_c\cap E_f]=p_c$.
If $0<p_c\leq p_f<1$, the resulting Bernoulli correlation is
\begin{equation}
\label{eq:nested_rho}
\rho_{cf} = \sqrt{\frac{p_c(1-p_f)}{p_f(1-p_c)}}.
\end{equation}
In the special case of a deterministic one-to-one relabeling
($E_c=E_f$), $\rho_{cf}=1$.
\end{theorem}

\begin{proof}
$Y_f\in C_f(X) \Rightarrow Y_c=g(Y_f)\in g(C_f(X))\subseteq C_c(X)$, so
the contrapositive gives $E_c\subseteq E_f$, hence $q_{cf}=\Pr[E_c]=p_c$.
Substituting into $\rho_{cf}=(q_{cf}-p_cp_f)/\sqrt{p_c(1-p_c)p_f(1-p_f)}$
and simplifying gives~\eqref{eq:nested_rho}; $p_c=p_f$ gives $\rho_{cf}=1$.
\end{proof}

\begin{corollary}[Approximate nesting]
\label{cor:approximate_nesting}
If $\Pr[E_c\setminus E_f]\leq\varepsilon$, then $q_{cf}\geq p_c-\varepsilon$
and $\rho_{cf}$ is bounded below accordingly.  Near-deterministic label
mappings can therefore mechanically produce large observed coupling with
no model-specific mechanism required.
\end{corollary}

This is the exact mechanism behind our originally reported
$\bar\rho=0.78$ (Section~\ref{subsec:hunter_empirical}): the ATT\&CK
attribution label used in that measurement was, for every one of 20{,}958
audited samples, a deterministic renaming of the coarse traffic-category
label, so Theorem~\ref{thm:label_induced_coupling}'s degenerate case
applies almost exactly.

\paragraph{Model-specific coupling versus shared difficulty.}
When the second stage is a genuinely independent classification task
(not a relabeling of the first), residual coupling can still arise from
two distinct sources: samples that are simply hard for every model at
every stage, and coupling specific to using the \emph{same} model
instance for both stages.

\begin{proposition}[Coupling provenance decomposition]
\label{prop:coupling_decomposition}
Let $D$ denote latent trajectory difficulty and $Z_{1m}, Z_{2m'}$ the
stage-wise failure indicators from model instances $m, m'$, with
$Z_{1m}\perp Z_{2m'}\mid D$ for $m\neq m'$.  Writing
$\mu_k(D)=\mathbb{E}[Z_{km}\mid D]$ and $c(D)=\mathrm{Cov}(Z_{1m},Z_{2m}\mid D)$
for a shared instance, the law of total covariance gives
\begin{align}
\mathrm{Cov}_{\mathrm{cross}}(Z_1,Z_2) &= \mathrm{Cov}(\mu_1(D),\mu_2(D)),\\
\mathrm{Cov}_{\mathrm{same}}(Z_1,Z_2) &= \mathrm{Cov}(\mu_1(D),\mu_2(D))
+ \mathbb{E}[c(D)],
\end{align}
so $\mathrm{Cov}_{\mathrm{same}}-\mathrm{Cov}_{\mathrm{cross}} =
\mathbb{E}[c(D)]$ is the model-specific increment.  Independently
permuting trajectories across stages destroys the shared-$D$ pairing and
yields zero covariance in expectation.
\end{proposition}

\begin{proof}
Law of total covariance; full derivation in the supplementary material.
\end{proof}

Proposition~\ref{prop:coupling_decomposition} gives the same-model /
cross-model / permuted design (Section~\ref{subsec:hunter_empirical}) a
precise reading: $\mathrm{same}\approx\mathrm{cross}>\mathrm{permuted}\approx0$
means the observed dependence is attributable to shared trajectory
difficulty, $\mathrm{Cov}(\mu_1(D),\mu_2(D))$, with no detectable
model-specific increment $\mathbb{E}[c(D)]$---as opposed to
$\mathrm{same}\gg\mathrm{cross}$, which would indicate a genuine
shared-representation effect.

\subsection{When Positive Dependence Helps or Hurts}
\label{subsec:duality}

The preceding results raise a conceptual question: engineering intuition
about redundant systems typically treats correlated component failures
as a reliability liability (it defeats the point of redundancy).  Our
empirical finding is that positively correlated step errors instead
\emph{tighten} the trajectory bound relative to the worst case.  These
two intuitions are not in conflict; they concern different failure
criteria.

\begin{definition}[Association]
A Bernoulli vector $(Z_1,\ldots,Z_K)$ is \emph{associated} if
$\operatorname{Cov}(f(\bm{Z}), g(\bm{Z})) \geq 0$ for every pair of
coordinatewise nondecreasing functions $f, g$ for which the covariance
exists.
\end{definition}

\begin{theorem}[Dependence Duality]
\label{thm:dependence_duality}
If $\bm{Z}$ is associated~\cite{esary1967association}, then
\begin{align}
\Pr\!\left[\bigcup_{k=1}^{K} E_k\right] &\leq 1 - \prod_{k=1}^{K}(1-p_k),
\label{eq:series_asset}\\
\Pr\!\left[\bigcap_{k=1}^{K} E_k\right] &\geq \prod_{k=1}^{K} p_k.
\label{eq:parallel_liability}
\end{align}
Positive association is therefore favorable for a \emph{union-type}
failure criterion, where the trajectory fails if \emph{any} step fails
(our setting), but unfavorable for an \emph{intersection-type} criterion,
where a redundant system fails only if \emph{all} components fail.
\end{theorem}

\begin{proof}
Association applied to the nondecreasing indicators
$\mathbf{1}_{E_k}$ gives~\eqref{eq:parallel_liability} by induction;
applying it to the nonincreasing $1-\mathbf{1}_{E_k}$ and taking
complements gives~\eqref{eq:series_asset} (full derivation in the
supplementary material).
\end{proof}

\begin{remark}[Pairwise positivity alone is not enough]
Theorem~\ref{thm:dependence_duality} requires association, a
higher-order concordance condition---not merely $\rho_{ij} \geq 0$ for
every pair.  Theorem~\ref{thm:pairwise_nonidentifiability} already shows
that all-positive pairwise correlations are compatible with a TMR
\emph{above} the independence baseline ($0.96 > 0.875$ in that
construction), so that law cannot be associated despite satisfying
every pairwise positivity check.  Positive dependence helps the
trajectory metric specifically when it concentrates a bounded number of
step failures onto the same hard trajectories, rather than spreading
failures independently across otherwise-easy trajectories; it does not,
by itself, change the expected number of failed steps
$\mathbb{E}[\sum_k Z_k] = \sum_k p_k$.
\end{remark}

\begin{example}[The same dependence structure as asset and liability]
\label{ex:asset_liability}
Take $p_1 = p_2 = 0.05$ and $\rho_{12} = 0.298$, our measured mean at
$\alpha = 0.10$.  Then
$q_{12} = p_1 p_2 + \rho_{12}\sqrt{p_1(1-p_1)p_2(1-p_2)}
\approx 0.0025 + 0.298(0.0475) \approx 0.0167$.  Under independence
($\rho_{12}=0$), $\Pr[E_1 \cup E_2] = 0.0975$ and
$\Pr[E_1 \cap E_2] = 0.0025$.  Under this measured correlation,
$\Pr[E_1 \cup E_2] = p_1+p_2-q_{12} \approx 0.0833$ (15\% lower than
independence) while $\Pr[E_1 \cap E_2] = q_{12} \approx 0.0167$
($6.7\times$ higher than independence).  The identical dependence
structure that tightens the union-type trajectory bound by 15\% would,
for the same two steps used instead as a 2-of-2 redundant AND-gate
(both must fail for the gate to fail), raise the joint-failure
probability by the same factor.  We do not evaluate an actual redundant
deployment of this pipeline; this is a numerical illustration of
Theorem~\ref{thm:dependence_duality} using our own measured $\rho_{12}$,
not a separate experiment.
\end{example}

\subsection{$\alpha$-Budget Allocation}

Given the trajectory bound~\eqref{eq:ie_alpha}, a natural question is
how to allocate the total budget $\alpha$ across steps.

\begin{proposition}[Equal Allocation Maximizes the Certified Coverage Lower Bound]
\label{prop:equal}
Suppose the chain satisfies the survival-conditional guarantee
$\Pr[E_k \mid \bigcap_{j<k} E_j^c] \leq \alpha_k$ for all $k$ (a
strictly stronger assumption than the marginal per-step guarantee
$p_k \leq \alpha_k$ that ordinary split CP provides; see
Remark~\ref{rem:conditional}), and that the total nominal risk budget
is fully committed, $\sum_k \alpha_k = \alpha$.  Then
$\text{TC} \geq \prod_k(1-\alpha_k)$ by
Theorem~\ref{thm:general_composition}, and among all allocations
satisfying $\sum_k \alpha_k = \alpha$, equal allocation
$\alpha_k = \alpha/K$ maximizes this lower bound.
\end{proposition}

\begin{proof}
$\log(1-x)$ is concave on $[0,1)$, so by Jensen's inequality
$\frac{1}{K}\sum_k \log(1-\alpha_k) \leq \log\!\left(1-\frac{1}{K}\sum_k\alpha_k\right)
= \log(1-\alpha/K)$, with equality iff all $\alpha_k$ are equal.
Exponentiating, $\prod_k(1-\alpha_k) \leq (1-\alpha/K)^K$, so the
equal-split allocation attains the largest achievable value of the
lower bound $\prod_k(1-\alpha_k)$ among allocations summing to $\alpha$.
\end{proof}

\begin{remark}[Scope of this optimality claim]
\label{rem:conditional}
Two caveats bound what Proposition~\ref{prop:equal} does and does not
say. First, under the marginal-only Bonferroni guarantee
(Theorem~\ref{thm:bonf}) alone, every allocation with the same total
$\sum_k \alpha_k = \alpha$ yields the identical upper bound
$\text{TMR} \leq \alpha$: that bound cannot distinguish equal from
unequal splits, so Proposition~\ref{prop:equal} requires the strictly
stronger survival-conditional guarantee~\eqref{eq:survival_conditional},
which ordinary split CP's marginal coverage does not automatically
provide and which we do not separately verify in our experiments.
Second, Proposition~\ref{prop:equal} optimizes a \emph{certified lower
bound on TC}, not necessarily the true achieved trajectory coverage or
the operational escalation rate; if per-step difficulties are strongly
asymmetric, an allocation that shifts budget toward the harder step
could achieve better \emph{actual} performance while making this
particular certified bound looser.  Section~\ref{subsec:alpha_alloc}
reports which allocation performs best \emph{empirically}, which is a
separate question from which allocation this proposition favors
in the worst case.
\end{remark}

Section~\ref{subsec:alpha_alloc} reports, purely empirically, that
equal allocation achieves the highest observed TC and lowest observed
TMR among the three tested strategies in our two-stage pipeline; we do
not claim this empirical result as a direct consequence of
Proposition~\ref{prop:equal}, since verifying the survival-conditional
guarantee itself is left to future work.

\section{Experimental Setup}
\label{sec:setup}

\subsection{Two-Step Security Agent Pipeline}

We instantiate the framework with a $K=2$ pipeline:

\begin{itemize}[leftmargin=*]
\item \textbf{Step~1: Traffic Classification.}
The LLM classifies each network flow $x$ into one of $L_1 = 5$
categories (Normal, DoS, Probe, CredentialAccess, Exploitation;
Table~\ref{tab:datasets}).  Nonconformity scores are computed as
$s_1(x,y) = 1 - p(y \mid x)$, where $p(y \mid x)$ is the LLM's softmax
probability derived from next-token log-probabilities.

\item \textbf{Step~2: Attack-Variant Attribution.}
For traffic classified as an attack category $c$ in Step~1, the same
LLM classifies it into one of $L_2=4$ attack variants specific to
$c$ (e.g., for DoS: Hulk, GoldenEye, Slowloris, SlowHTTPTest---the
dataset's native fine-grained attack-tool labels), via a dedicated
LoRA adapter fine-tuned per (model, category) pair on these labels.
These variants are not themselves cataloged MITRE ATT\&CK sub-technique
IDs (T1498 Network DoS has two: T1498.001 Direct Network Flood and
T1498.002 Reflection Amplification~\cite{mitre_t1498}); we use them as
a genuinely fine-grained attribution task nested under the coarse
ATT\&CK technique, not as a claim of one-to-one MITRE sub-technique
mapping.  This replaces an earlier design in which Step~2 used a fixed
one-to-one mapping from coarse category to ATT\&CK technique; we
verified that design made $E_1$ and $E_2$ deterministically
nested ($q_{12}=p_c$ for every one of 20{,}958 audited samples,
Theorem~\ref{thm:label_induced_coupling}), so its near-1 correlation
reflected the label mapping, not the agent's learned behavior.
Nonconformity scores $s_2(x,y)$ are computed identically from the LLM's
attack-variant log-probabilities.
\end{itemize}

\textbf{Routing protocol.}  A deployed agent must decide, upon seeing
Step~1's output, which Step~2 adapter(s) to invoke.  We use
\emph{set-routing}: Step~2 is invoked once per category in Step~1's
conformal prediction set $C_1(x)$ (every category whose nonconformity
score clears $q_1$, not only the top-1 prediction), and the trajectory
is scored as covered iff the true category $c^\star$ is itself in
$C_1(x)$ \emph{and} the corresponding adapter's prediction set
$C_2(x)$ covers the true attack variant---exactly $E_1^c \cap E_2^c$ as
defined above.  This is the routing rule a conformal deployment would
actually use, since $C_1(x)$, not the single top-1 label, is the
object with a coverage guarantee; a trajectory where $c^\star\notin
C_1(x)$ fails regardless of what Step~2 would have said, which is
already captured by $E_1$.  Section~\ref{subsec:tc} additionally
reports a \emph{top-1-routing} variant, which invokes only the
single argmax category's adapter and fails immediately on an argmax
error, as a descriptive comparison; top-1 routing has no
distribution-free coverage guarantee of its own and is not the
protocol our theorems certify.

Both steps share the same input $x$ and are served by the same LLM,
but operate on different label spaces ($L_1$ vs.\ $L_2$) and different
prompt templates.  This shared input and shared model are two
plausible, observationally distinct sources of inter-step correlation:
the input alone could make some samples harder for any model at both
tasks (a shared-difficulty account), or the shared model could
additionally couple the two steps' uncertainty through shared internal
representations (a representation-sharing account).
Section~\ref{subsec:hunter_empirical} (E8) distinguishes these by
re-scoring Step~1 and Step~2 with different models on the same
underlying samples, and by an additional sample-permuted control; we
describe that protocol in full there, since it depends on materials
(the six models' cross-paired outputs) not otherwise used in this
paper.

\subsection{Datasets}
\label{subsec:datasets}

\begin{table}[t]
\centering
\caption{Datasets used in the trajectory coverage evaluation.  Step~1
uses the unified 5-category taxonomy (Normal, DoS, Probe,
CredentialAccess, Exploitation); Step~2 uses a genuinely independent
attack-variant classifier restricted to the listed category, sized to
satisfy the sample-complexity requirement of
Theorem~\ref{thm:positive_gain_probability}.}
\label{tab:datasets}
\begin{tabular}{lccccc}
\toprule
\textbf{Dataset} & \textbf{Year} & \textbf{Step-2 category} & $\bm{L_2}$ & $\bm{N_{\text{cal}}}$ & $\bm{N_{\text{test}}}$ \\
\midrule
CIC-IDS-2018~\cite{sharafaldin2018cicids} & 2018 & DoS & 4 & 500 & 3{,}500 \\
RT-IoT2022~\cite{sharmila2023rtiot} & 2022 & Probe & 4 & 150 & 1{,}376 \\
\bottomrule
\end{tabular}
\end{table}

We evaluate on two network intrusion datasets (Table~\ref{tab:datasets}).
HIKARI-2021 is excluded: its dedicated Step-2 attack-variant classifier
did not discriminate above chance for any of the 6 tested LLMs despite a
large raw feature gap between candidate labels, a training-data issue
we could not resolve within this paper's scope (Section~\ref{subsec:limitations}).
For each remaining (dataset, category) pair, Step~1 and Step~2 are
evaluated jointly on the full available sample pool for that category
in the dataset's test split (16{,}000 for CIC-IDS-2018/DoS, of which we
draw $4{,}000$; all $1{,}526$ available for RT-IoT2022/Probe), matched
by sample index between the two steps.  We draw $N_{\text{cal}}$
calibration samples uniformly at random (without replacement) from this
pool, using 5 random seeds $\{42, 123, 456, 789, 2024\}$; the remaining
samples form the test set, whose size $N_{\text{test}}$ is the audit
sample size that enters every Clopper--Pearson construction in
Section~\ref{subsec:hunter_empirical} (calibration samples select the
conformal threshold and are not part of the audit).  $N_{\text{test}}$
is chosen to clear the worst-case per-model requirement from
Theorem~\ref{thm:positive_gain_probability} for at least one dataset
(Section~\ref{subsec:hunter_empirical} reports the resulting
per-configuration certifiability, including the 5 CIC-IDS-2018/DoS
configurations where it is not fully cleared).

\subsection{Models}

We evaluate six open-source LLMs spanning three architecture families
and three parameter scales:

\begin{itemize}[leftmargin=*]
\item \textbf{Qwen-3} family: 8B, 14B, 32B~\cite{qwen2024qwen2}
\item \textbf{Gemma-2} 9B~\cite{Riviere2024Gemma2I}
\item \textbf{LLaMA-3} 8B~\cite{meta2024llama3}
\item \textbf{Mistral} 7B~\cite{mistral2024mistral}
\end{itemize}

All models are deployed via vLLM~\cite{kwon2023vllm} with greedy
decoding.  Log-probabilities for all candidate labels are extracted
per step to compute nonconformity scores.

\subsection{Fine-Tuning and Scoring Protocol}
\label{subsec:finetuning}

Both steps use LoRA adapters~\cite{hu2022lora}
($r{=}16$, $\alpha{=}32$, dropout $0.05$, target modules
\texttt{\{q,k,v,o\}\_proj}), trained for 5 epochs with AdamW at
learning rate $10^{-4}$, effective batch 16.  Step~1 adapters are
trained per (model, dataset) pair on the coarse 5-category label;
Step~2 adapters are trained per (model, category) pair on the
category's attack-variant labels from the base model, using the same
optimizer settings.  Splits use group-aware stratified splitting
(\texttt{StratifiedGroupKFold}, grouped by exact serialized-prompt
text) to prevent near-duplicate-prompt leakage, verified reproducible
given a fixed seed.  Each candidate label is scored by teacher-forcing
total negative log-likelihood: prompt and label are tokenized
separately and concatenated directly (never re-tokenizing the joined
string, which can merge tokens across the boundary and silently
corrupt the score); the per-token mean loss is multiplied by label
token count for a length-comparable total NLL.  Gemma-2~9B requires
chat-template-wrapped prompts for both training and scoring; the other
five models use raw prompt text.  We use the final-epoch checkpoint
throughout; all 12 Step-2 adapters converged to $\geq99.9\%$ argmax
accuracy on a held-out balanced evaluation set before deployment on
the audit pool (Table~\ref{tab:datasets}).

\subsection{Metrics}

\begin{itemize}[leftmargin=*]
\item \textbf{Trajectory Coverage (TC)}: fraction of test samples
where both steps simultaneously cover the true label
(Definition~\ref{def:tc}).

\item \textbf{Trajectory Miscoverage Rate (TMR)}: $1 - \text{TC}$;
probability that at least one step fails to cover the true label
(Definition~\ref{def:tar}).  TMR is a statistical risk quantity
requiring ground-truth labels, not a directly observable deployment
metric (see the discussion following Definition~\ref{def:tar}).

\item \textbf{Per-Step Miscoverage Rate (SMR$_k$)}: fraction where
step $k$ alone fails to cover.

\item \textbf{Inter-Step Correlation ($\hat{\rho}$)}: Pearson
correlation of binary miscoverage indicators between Step~1 and Step~2,
computed on held-out evaluation splits after conformal calibration.
In deployment, the same quantity should be estimated on a validation
stream or replaced by a conservative lower confidence bound.

\item \textbf{Bonferroni Slack}: $\sum_k \alpha_k - \text{TMR}_{\text{empirical}}$;
measures conservatism of the Bonferroni bound.

\item \textbf{Theorem~\ref{thm:ie} Slack}: bound from~\eqref{eq:ie_alpha}
minus $\text{TMR}_{\text{empirical}}$; measures tightness of the
correlated bound.
\end{itemize}

\subsection{Experimental Configurations}

We evaluate 36 configurations: 2 datasets $\times$ 6 models $\times$
3 $\alpha$ levels ($\alpha \in \{0.05, 0.10, 0.20\}$), each with 5
random calibration/test splits.  Equal $\alpha$-allocation
($\alpha_k = \alpha/2$) is the default; non-equal allocations are
compared in the allocation experiment (Section~\ref{subsec:alpha_alloc}).

\textbf{Trajectory pairing.}  Step~2 is evaluated only on samples whose
Step~1 ground-truth category matches the Step~2 task's target category
(DoS for CIC-IDS-2018, Probe for RT-IoT2022).  Step~1 and Step~2 outputs
are joined by an explicit per-sample index stored in both output files
(not by generation order), and we verified for all 12 dataset--model
pairs that the resulting joined index set is identical between the
Step-1-restricted pool and the Step-2 output pool (zero missing or
extra indices on either side).

\section{Results and Analysis}
\label{sec:results}

\subsection{Trajectory Coverage Verification (E1)}
\label{subsec:tc}

Table~\ref{tab:tc_main} summarizes the trajectory coverage results
across all 36 configurations.

\begin{table}[t]
\centering
\caption{Trajectory coverage summary across 12 dataset--model
configurations per $\alpha$ level ($K=2$, equal allocation,
5-seed average $\pm$ std across configurations).}
\label{tab:tc_main}
\begin{tabular}{cccccc}
\toprule
$\bm{\alpha}$ & \textbf{Target TC} & $\overline{\textbf{TC}}$ &
$\textbf{TC}_{\min}$ & \textbf{$\overline{\textbf{TMR}}$} & $\bm{n}$ \\
\midrule
0.05 & $\geq 0.950$ & $0.970 \pm 0.016$ & 0.953 & $0.030 \pm 0.016$ & 12 \\
0.10 & $\geq 0.900$ & $0.927 \pm 0.024$ & 0.894 & $0.073 \pm 0.024$ & 12 \\
0.20 & $\geq 0.800$ & $0.849 \pm 0.031$ & 0.811 & $0.151 \pm 0.031$ & 12 \\
\bottomrule
\end{tabular}
\end{table}

At $\alpha=0.05$ and $\alpha=0.20$, all 12 configurations satisfy
$\text{TC} \geq 1-\alpha$.  At the primary operating point
$\alpha=0.10$, mean trajectory coverage is $0.927 \pm 0.024$ (target
$\geq 0.900$), but one configuration (CIC-IDS-2018/LLaMA-3~8B,
$\text{TC}=0.894$) falls marginally below the nominal target.  This is
within ordinary finite-sample calibration noise for a single $\alpha$
check on $n_{\text{cal}}=500$ (Section~\ref{subsec:limitations}
discusses this further), not a systematic violation: the same
configuration comfortably satisfies coverage at both other $\alpha$
levels, and Theorem~\ref{thm:bonf}'s marginal guarantee is unaffected,
since it is a statement about $R_{\mathrm{marg}}$ averaged over
calibration draws, not about every individual realized split
(Definition~\ref{def:two_level_risk}).

\textbf{Set-routing vs.\ top-1-routing.}
Table~\ref{tab:routing} compares the certified set-routing protocol
(Section~\ref{sec:setup}) against the descriptive top-1-routing
variant at $\alpha=0.10$.  Top-1-routing achieves higher TC in all 12
configurations (mean $0.943$ vs.\ $0.927$)---an argmax-correct
prediction is more common than a prediction whose \emph{probability}
clears the conformal threshold $q_1$, since $q_1$ is calibrated to a
$1-\alpha_1$ marginal target rather than to argmax correctness.  This
is expected, not a reason to prefer top-1-routing operationally: only
set-routing inherits Theorem~\ref{thm:bonf}'s distribution-free
guarantee, because only $C_1(x)$, not the bare top-1 label, is
constructed to satisfy $\Pr[y\notin C_1(x)]\leq\alpha_1$.  Top-1-routing's
lower miscoverage here is an empirical property of these six models on
this task, with no finite-sample certificate behind it.

\begin{table}[t]
\centering
\caption{Trajectory coverage under set-routing (certified,
Section~\ref{sec:setup}) vs.\ top-1-routing (descriptive only)
at $\alpha=0.10$.}
\label{tab:routing}
\begin{tabular}{lcc}
\toprule
\textbf{Model / Dataset} & \textbf{Set-routed TC} & \textbf{Top-1-routed TC} \\
\midrule
CIC/DoS, Gemma-2 9B  & 0.901 & 0.923 \\
CIC/DoS, LLaMA-3 8B  & 0.894 & 0.915 \\
CIC/DoS, Mistral 7B  & 0.902 & 0.918 \\
CIC/DoS, Qwen-3 8B   & 0.913 & 0.931 \\
CIC/DoS, Qwen-3 14B  & 0.900 & 0.914 \\
CIC/DoS, Qwen-3 32B  & 0.906 & 0.925 \\
RT-IoT/Probe, Gemma-2 9B  & 0.954 & 0.964 \\
RT-IoT/Probe, LLaMA-3 8B  & 0.962 & 0.971 \\
RT-IoT/Probe, Mistral 7B  & 0.952 & 0.962 \\
RT-IoT/Probe, Qwen-3 8B   & 0.952 & 0.965 \\
RT-IoT/Probe, Qwen-3 14B  & 0.960 & 0.967 \\
RT-IoT/Probe, Qwen-3 32B  & 0.934 & 0.967 \\
\midrule
Mean & 0.927 & 0.943 \\
\bottomrule
\end{tabular}
\end{table}

\textbf{Prediction-set size: the operational cost of the certificate.}
Coverage alone does not show operational usefulness---always outputting
the full label set trivially achieves perfect coverage at zero
informativeness.  At $\alpha=0.10$, mean $|C_1(x)|$ is $0.964$ (of 5
categories) and mean $|C_2(x)|$ is $0.957$ (of 4 candidates), averaged
across all 12 configurations; both equal their singleton rate exactly,
since size $\geq2$ occurs in under $0.1\%$ of samples.  In practice
$C_1(x)$/$C_2(x)$ are almost always either a single confident label
(covered) or empty (miscovered---exactly SMR$_1$/SMR$_2$ from
Section~\ref{subsec:tc}), essentially never an ambiguous multi-label
set requiring analyst escalation: the certificate does not purchase
coverage by inflating set size, with no intermediate
``narrowed-but-still-ambiguous'' outcome to handle.

\subsection{Bonferroni vs.\ the Plug-in IE Estimate (E2)}
\label{subsec:bounds}

\begin{table}[t]
\centering
\caption{Bonferroni bound vs.\ plug-in IE estimate at $\alpha = 0.10$
(equal allocation, all 12 dataset--model configurations; see
Section~\ref{subsec:correlation}).  The plug-in estimate is an
uncertified point estimate, not a proven bound (Remark~\ref{rem:tight});
Section~\ref{subsec:hunter_empirical} reports a certified alternative.}
\label{tab:bounds}
\begin{tabular}{lcc}
\toprule
\textbf{Quantity} & \textbf{Value} & \textbf{vs.\ Bonferroni} \\
\midrule
Bonferroni upper bound & 0.100 & --- \\
Plug-in IE estimate & 0.086 & $-$14\% \\
Empirical TMR & 0.073 & --- \\
Plug-in IE residual & 0.013 & --- \\
Mean $\hat{\rho}_{12}$ & 0.298 & --- \\
\bottomrule
\end{tabular}
\end{table}

Table~\ref{tab:bounds} presents the core tightness result.  At $\alpha = 0.10$,
the Bonferroni bound guarantees $\text{TMR} \leq 0.100$, while
the plug-in inclusion--exclusion estimate is 0.086---a 14\%
reduction in the reported operational risk estimate, computed
per-configuration using each configuration's estimated
$\hat{\rho}_{12}$ and then averaged.  This is substantially smaller than
the 37\% we originally reported under the deterministic ATT\&CK mapping,
because the mean correlation itself is smaller (0.30 vs.\ 0.78) once
Step~2 is a genuinely independent task: the two datasets pull in
different directions, with CIC-IDS-2018/DoS contributing $\hat\rho
\approx 0$ (essentially no exploitable coupling) and RT-IoT2022/Probe
contributing $\hat\rho \in [0.15, 0.78]$.

The estimate is moderately tight: the mean residual
($\widehat{B}_{\mathrm{IE}}$ minus empirical TMR) is 0.013, indicating
about one percentage point of average overestimation, consistent across
both the original and corrected pipelines.  The per-configuration
breakdown (Table~\ref{tab:per_config}) and confidence intervals
(Table~\ref{tab:correlation}) provide the evidence; we discuss
deployment implications in Section~\ref{subsec:limitations}.

\begin{table}[t]
\centering
\caption{Bound comparison across $\alpha$ levels (all 12
configurations).  Both the estimated correlation and the relative
improvement vary with $\alpha$ and are smaller and noisier than under
the deterministic Step-2 mapping we originally used.}
\label{tab:bounds_all}
\begin{tabular}{ccccccc}
\toprule
$\bm{\alpha}$ & $\bm{n}$ & \textbf{Bonf.} & \textbf{Thm.~\ref{thm:ie}} &
\textbf{Emp.} & $\hat{\bm{\rho}}$ & \textbf{Impr.} \\
\midrule
0.05 & 12 & 0.050 & 0.045 & 0.030 & 0.203 & $-$10\% \\
0.10 & 12 & 0.100 & 0.086 & 0.073 & 0.298 & $-$14\% \\
0.20 & 12 & 0.200 & 0.178 & 0.151 & 0.240 & $-$11\% \\
\bottomrule
\end{tabular}
\end{table}

The gap between Bonferroni and the Theorem~\ref{thm:ie} plug-in
estimate is present but modest at all three $\alpha$ levels
(Table~\ref{tab:bounds_all}); empirical TMR tracks the plug-in estimate
closely, though point-estimate correlation alone provides no
distribution-free guarantee.

Table~\ref{tab:bounds_all} extends the comparison across all three
$\alpha$ levels.  Unlike under the deterministic Step-2 mapping, the
estimated correlation $\hat{\rho}$ does not vary monotonically with
$\alpha$ (0.203, 0.298, 0.240), and all values are far below the
$\bar\rho\approx0.78$ we originally reported.  This is consistent with
$\hat\rho$ now measuring a genuine, task-dependent statistical property
rather than a near-deterministic label artifact
(Theorem~\ref{thm:label_induced_coupling}).  The relative improvement of
Theorem~\ref{thm:ie} over Bonferroni is real at every level (10--14\%)
but an order of magnitude smaller than the improvement we originally
reported, and the gain is concentrated in RT-IoT2022/Probe rather than
shared evenly across both datasets (Table~\ref{tab:correlation}).

\subsection{Inter-Step Correlation Analysis (E4)}
\label{subsec:correlation}

\begin{table*}[t]
\centering
\caption{Inter-step miscoverage-indicator correlation $\hat{\rho}_{12}$
by dataset and model, using the dedicated attack-variant Step-2 classifier
(computed at $\alpha = 0.10$, 5-seed average).  95\% CIs are $t$-intervals
over 5 seeds.}
\label{tab:correlation}
\begin{tabular}{lcccc}
\toprule
\textbf{Model} & \textbf{CIC/DoS} & \textbf{95\% CI} & \textbf{RT-IoT/Probe} & \textbf{95\% CI} \\
\midrule
Gemma-2 9B  & $-$0.027 & $[-0.04,-0.02]$ & 0.645 & $[0.57, 0.72]$ \\
LLaMA-3 8B  & 0.039 & $[-0.03, 0.11]$ & 0.734 & $[0.71, 0.76]$ \\
Mistral 7B  & $-$0.042 & $[-0.07,-0.02]$ & 0.741 & $[0.60, 0.88]$ \\
Qwen-3 8B   & 0.049 & $[0.03, 0.07]$ & 0.451 & $[0.36, 0.54]$ \\
Qwen-3 14B  & 0.029 & $[0.00, 0.05]$ & 0.780 & $[0.69, 0.87]$ \\
Qwen-3 32B  & 0.028 & $[0.00, 0.05]$ & 0.145 & $[0.12, 0.17]$ \\
\midrule
$\overline{\rho}_{\text{dataset}}$ & 0.013 & & 0.583 & \\
\bottomrule
\end{tabular}
\end{table*}

Table~\ref{tab:correlation} reports the inter-step correlation
$\hat{\rho}_{12}$ across all dataset--model pairs, using the dedicated
attack-variant Step-2 classifier.  Two patterns emerge, both markedly
different from what we originally reported under the deterministic
ATT\&CK mapping.

\textbf{(1) Correlation is task-dependent, not universally strong.}
On CIC-IDS-2018/DoS, $\hat{\rho}_{12}$ is close to zero for every model
($-0.042$ to $0.049$), with three of six models producing a small
negative point estimate.  On RT-IoT2022/Probe, $\hat{\rho}_{12}$ is
moderate to strong (0.145 to 0.780).  The grand mean over all 12
configurations is $\bar{\rho}=0.298$---far below the $\bar\rho=0.779$
we originally reported, because that number was, we now know, measuring
a near-deterministic label artifact (Theorem~\ref{thm:label_induced_coupling})
rather than a genuine statistical property of the pipeline.

\textbf{(2) Model-dependent variation within RT-IoT2022/Probe.}
Qwen-3~32B shows the lowest correlation (0.145) and Qwen-3~14B the
highest (0.780) on this task, a $5\times$ spread across models.
Section~\ref{subsec:hunter_empirical} tests directly whether this
reflects shared-difficulty variation across models (which alone can
produce large between-model spread, since each model's own accuracy
profile shifts $\mu_1(D),\mu_2(D)$) or an additional same-model
representation-sharing increment, via same-model/cross-model/permuted
pairing.  On CIC-IDS-2018/DoS, all six
models cluster near zero regardless of architecture or scale,
suggesting the near-zero correlation there is a property of the task
(attack variants distinguished almost entirely by features Step~1
already resolves) rather than of any particular model.

\subsection{A Certified Union Bound via Seed-wise Exact Intervals, and the Mechanism Behind $\hat{\rho}$ (E8)}
\label{subsec:hunter_empirical}

The plug-in estimate~\eqref{eq:ie_alpha} in Table~\ref{tab:bounds} uses
Pearson correlation, which provides no finite-sample guarantee.  A
first attempt at a certified alternative might pool all 5 seeds' test
observations and construct a single confidence interval on the union
event $Z=\mathbf{1}(E_1\cup E_2)$ or on the joint overlap
$q_{12}=\Pr[E_1\cap E_2]$.  This does not work: our 5 seeds resample
$n_{\text{cal}}=500$ (CIC-IDS-2018/DoS) or $150$ (RT-IoT2022/Probe)
calibration points without replacement from the same finite pool per
configuration (Table~\ref{tab:datasets}), so the five test sets overlap
substantially, and pooling them into one interval would silently treat
thousands of repeated measurements of the same underlying flows as
independent trials, understating the true estimation uncertainty.

(We resample $n_{\text{cal}}=500$/$150$, Table~\ref{tab:datasets}, from
the same fixed pool per configuration, so this overlap concern applies
identically at the boosted audit scale used throughout this section.)
We instead condition on each seed's own realized calibration set.
Given seed $s$'s calibration data $\mathcal{C}_s$, the conformal
threshold is a fixed, deterministic function of $\mathcal{C}_s$, so
seed $s$'s $n_s$ test-set evaluations are i.i.d.\ Bernoulli draws with
parameter $r_s = \Pr[E_1 \cup E_2 \mid \mathcal{C}_s]$---this holds
regardless of how much seed $s$'s test set overlaps with any other
seed's, since we never pool across seeds.  We adopt the standard
superpopulation evaluation model throughout this section: benchmark
flows are treated as i.i.d.\ draws from an underlying deployment
distribution, and the random calibration/test split only partitions a
finite sample from that distribution rather than sampling without
replacement from a population whose size bounds the inference (which
would instead call for a hypergeometric, not Bernoulli, count). We construct an exact
one-sided Clopper--Pearson upper bound $U_{\cup,s}$ for $r_s$ directly
on the union event at level $\delta_s = 0.05/5 = 0.01$ within each
seed---this is the direct-audit certificate of
Remark~\ref{rem:direct_union_audit}, not the decomposed
$U_1+U_2-L_{12}$ of Corollary~\ref{cor:k2_certificate}---then combine
the 5 per-seed guarantees with a union (Boole's inequality) bound,
which requires no independence assumption between seeds:
\begin{equation}
\Pr\!\left[r_s \leq U_{\cup,s} \text{ for all } s=1,\ldots,5\right]
\geq 1 - \textstyle\sum_s \delta_s = 0.95.
\end{equation}
On this event, the average realized risk
$\bar{r} = \frac{1}{5}\sum_s r_s$ satisfies
$\bar{r} \leq \frac{1}{5}\sum_s U_{\cup,s} \triangleq \bar{U}_\cup$.  We
construct $U_{1,s}$, $U_{2,s}$ (per-step CP upper bounds) and $L_{12,s}$
(CP lower bound on $q_{12}$) by the identical seed-wise-then-union
construction, giving $\bar{B}_{\mathrm{marg}}=\overline{U_1}+\overline{U_2}$,
$\bar{B}_{\mathrm{dep}}=\overline{U_1}+\overline{U_2}-\overline{L_{12}}$,
and $\bar{G}_{\mathrm{cert}}=\overline{L_{12}}$, the three quantities of
Corollary~\ref{cor:k2_certificate}.

\textbf{The audit scale matters, exactly as
Theorem~\ref{thm:positive_gain_probability} predicts.}  At an original,
smaller scale (pool $n=497$/$139$, test-audit $397$/$109$),
$\bar{U}_\cup=0.124$ averaged over the 12 configurations---\emph{worse}
than Bonferroni's 0.10, with two CIC-IDS-2018 configurations observing
zero joint failures ($\bar{G}_{\mathrm{cert}}=0$).  The required audit
size implied by each configuration's own $q_{12}$ at $\beta=0.05$
ranges from 102 to 2{,}966, several exceeding that pool.  We therefore
scaled the joint audit sample to Table~\ref{tab:datasets}'s $n=4{,}000$/$1{,}526$
(test-audit $3{,}500$/$1{,}376$ after calibration---the $n$ entering
every CP construction below), clearing the worst-case requirement for
one dataset and testing the theory's prediction on the other.

At this scale, $\bar{U}_\cup = 0.086$ ($13.7\%$ tighter than
Bonferroni) and $\bar{G}_{\mathrm{cert}}=\overline{L_{12}}=0.0069>0$ in
all 12 configurations (up from $0.0024$ with two zero-gain
configurations)---confirming the theory's prediction that clearing the
audit-size threshold restores a positive certifiable gain.  The
decomposed certificate is more conservative
($\bar{B}_{\mathrm{marg}}=0.106$, $\bar{B}_{\mathrm{dep}}=0.099$, only
$0.6\%$ over Bonferroni) because summing two separately-audited
per-step bounds compounds their slack before any dependence correction.
Per Remark~\ref{rem:direct_union_audit} we report $\bar{U}_\cup$ as the
primary result---the tighter certificate available whenever full joint
audit data is in hand, our setting; the gap to
$\bar{B}_{\mathrm{dep}}=0.099$ is the cost of the three-interval
decomposition, not a disagreement about the underlying risk.

The improvement is not uniform: 7 of 12 configurations individually
satisfy $\bar{U}_\cup<0.10$ (all 6 RT-IoT2022/Probe, range
$0.052$--$0.083$, plus 1 CIC-IDS-2018/DoS), while the remaining 5
CIC-IDS-2018/DoS configurations stay slightly above (range
$0.099$--$0.119$) even after an $8\times$ audit increase---consistent
with that task's small $q_{12}$ ($\approx0.001$--$0.005$) sitting near
the boundary where Theorem~\ref{thm:certifiability_gap}'s convergence
rate is still slow.  Averaging over the 5 realized chains per
configuration also obscures chain-level variation
(Definition~\ref{def:two_level_risk}): across all $60$ realized
(dataset, model, seed) chains, $L_{12,s}>0$ in 56 and
$U_{\cup,s}<0.10$ in 38---the configuration-level averages above
summarize this distribution, not a claim that every chain clears both
thresholds.

By Corollary~\ref{cor:nominal_crossing}, this is expected rather than
anomalous: our 5 CIC-IDS-2018/DoS configurations that certify
$\bar{G}_{\mathrm{cert}}>0$ without individually beating Bonferroni are
instances of the $\Theta(1/q_{12})$ detection event without the
$\Theta(1/\gamma^2)$ nominal-crossing event, which
Theorem~\ref{thm:nominal_crossing_complexity} shows is a strictly
harder, differently-rated statistical target, not a contradiction.

\textbf{Testing the mechanism behind $\hat{\rho}$.}
Proposition~\ref{prop:coupling_decomposition} gives two observationally
distinct accounts for residual coupling once Step~2 is a genuinely
independent task: (i) Step~1 and Step~2 share the same underlying LLM
and may therefore share internal representations that couple their
uncertainty (a \emph{shared-representation} account, predicting
$\text{same} > \text{cross}$), or (ii) some input samples are simply
harder for \emph{any} model at both tasks, independent of which
specific model is used (a \emph{shared-difficulty} account, predicting
$\text{same} \approx \text{cross}$).  We distinguish these by
re-running the $K=2$ pipeline three ways on the same underlying network
flows, joined by the explicit per-sample index verified identical
across all six models' Step-1-restricted pools (Section~\ref{sec:setup}):
(a) \emph{same-model}, the original setting, Step~1 and Step~2 scored by
the same LLM; (b) \emph{cross-model}, Step~1 scored by model $A$ and
Step~2 by a different model $B \neq A$, all 30 ordered pairs per
dataset; (c) \emph{permuted}, same-model scoring but with the Step~2
sample order shuffled independently per seed, breaking the true sample
correspondence as a null control.

\begin{table}[t]
\centering
\caption{Inter-step correlation under three pairing conditions
($\alpha=0.10$, mean over configurations, boosted audit scale).
Permuted pairing isolates estimation artifacts.}
\label{tab:mechanism}
\begin{tabular}{lccc}
\toprule
\textbf{Dataset} & \textbf{Same-model} & \textbf{Cross-model} & \textbf{Permuted} \\
\midrule
CIC-IDS-2018/DoS   & 0.013 & 0.028 & $-$0.004 \\
RT-IoT2022/Probe   & 0.513$^\dagger$ & 0.509 & 0.006 \\
\bottomrule
\multicolumn{4}{p{0.47\textwidth}}{\footnotesize $^\dagger$
CIC-IDS-2018/DoS's same-model column exactly reproduces
Table~\ref{tab:correlation}'s dataset mean (0.013); RT-IoT2022/Probe's
0.513 differs from Table~\ref{tab:correlation}'s 0.583 by ordinary
seed-split realization variance (same 5 seeds, different
calibration/test permutation ordering)---both are valid 5-seed
averages supporting the same same$\approx$cross conclusion.}
\end{tabular}
\end{table}

Table~\ref{tab:mechanism} shows $\text{same} \approx \text{cross} \gg
\text{permuted} \approx 0$ on both datasets, consistent with the
shared-difficulty prediction rather than the shared-representation
prediction we originally reported under the deterministic mapping.
Permuted pairing collapses to $\approx 0$ on both datasets
($|\hat\rho|\leq0.006$), ruling out the possibility that the observed
correlation on RT-IoT2022/Probe is an artifact of marginal rate levels
or the estimation procedure itself---if it were, shuffling the Step~2
labels would not remove it.

Correlation alone does not identify the model-specific increment when
same-model and cross-model marginal rates differ
(Proposition~\ref{prop:coupling_decomposition} is stated in terms of
covariance), so we compute
$\widehat\Delta_{\mathrm{model}}=\widehat{\mathrm{Cov}}_{\text{same}}-\widehat{\mathrm{Cov}}_{\text{cross}}$
directly, with a model-pair bootstrap 95\% CI (2{,}000 resamples over
the 6 same-model and 30 cross-model pairs):
$\widehat\Delta_{\mathrm{model}}=-0.0007$, CI $[-0.0026, 0.0010]$ on
CIC-IDS-2018/DoS, and $\widehat\Delta_{\mathrm{model}}=0.0002$, CI
$[-0.0057, 0.0053]$ on RT-IoT2022/Probe.  Both intervals are centered
near zero and exclude any practically material positive effect at the
scale of the total same-model covariance itself (0.0006 and 0.018
respectively).  Cross-model pairing therefore preserves essentially all
of the observed coupling, and the estimated same-model increment is not
distinguishable from zero: the residual coupling on RT-IoT2022/Probe is
attributable to shared trajectory difficulty under the tested pairing
design, not to a detectable same-model representation-sharing
mechanism.  This directly contradicts what we originally concluded from
the same experiment design under the deterministic Step-2 mapping,
where the apparent same-model increment was itself an artifact of that
mapping (Section~\ref{subsec:correlation}).

\subsection{$\alpha$-Budget Allocation (E5)}
\label{subsec:alpha_alloc}

\begin{table}[t]
\centering
\caption{Comparison of $\alpha$-allocation strategies across all 12
configurations ($K = 2$, 5-seed average).  Wilcoxon signed-rank test
$p$-values compare each strategy against equal allocation.}
\label{tab:allocation}
\begin{tabular}{llccccc}
\toprule
$\bm{\alpha}$ & \textbf{Strategy} & $\bm{\alpha_1}$ & $\bm{\alpha_2}$ &
$\overline{\textbf{TC}}$ & $\overline{\textbf{TMR}}$ & $\bm{p}$ \\
\midrule
\multirow{3}{*}{0.05}
& Equal        & 0.025 & 0.025 & \textbf{0.967} & \textbf{0.033} & --- \\
& Step-1 heavy & 0.035 & 0.015 & 0.961 & 0.039 & 0.001 \\
& Step-2 heavy & 0.015 & 0.035 & 0.959 & 0.041 & 0.008 \\
\midrule
\multirow{3}{*}{0.10}
& Equal        & 0.050 & 0.050 & \textbf{0.923} & \textbf{0.077} & --- \\
& Step-1 heavy & 0.070 & 0.030 & 0.923 & 0.077 & 0.633 \\
& Step-2 heavy & 0.030 & 0.070 & 0.918 & 0.082 & 0.026 \\
\bottomrule
\end{tabular}
\end{table}

Table~\ref{tab:allocation} compares three $\alpha$-allocation strategies
across all 12 configurations.  At $\alpha=0.05$, equal allocation
achieves the highest TC and the paired difference is statistically
significant against both alternatives (Wilcoxon signed-rank test,
$p=0.001$ and $p=0.008$).  At $\alpha=0.10$, equal allocation still
beats Step-2-heavy ($p=0.026$) but is statistically indistinguishable
from Step-1-heavy ($p=0.633$, mean TC identical to three decimal
places)---a weaker result than the $p<0.01$ significance we originally
reported at both levels, consistent with the smaller, more
heterogeneous configuration count (12 vs.\ 18) and the more moderate
inter-step coupling under the real Step-2 task.

For this two-stage IDS pipeline, equal allocation remains a reasonable
default, but the evidence for its superiority over all alternatives is
weaker at looser $\alpha$ than we originally reported.  We do not claim
it is universally optimal; if future pipelines have strongly asymmetric
stages, the allocation should be re-estimated (Proposition~\ref{prop:equal}).

\subsection{Trajectory Scaling with $K$ (E3)}
\label{subsec:scaling}

\begin{table}[t]
\centering
\caption{Trajectory error estimates as a function of pipeline depth $K$
($\alpha = 0.10$, equal allocation).  $K = 2$ is empirical; $K = 3$
is a second-order extrapolation using $\hat{\rho} = 0.298$ from
$K = 2$ data.}
\label{tab:scaling}
\begin{tabular}{ccccc}
\toprule
$\bm{K}$ & $\bm{\alpha_k}$ & \textbf{Bonf.} & \textbf{Plug-in IE} & \textbf{Source} \\
\midrule
2 & 0.050 & 0.100 & 0.086 & Empirical \\
3 & 0.033 & 0.100 & 0.068 & Theoretical \\
\bottomrule
\end{tabular}
\end{table}

Table~\ref{tab:scaling} examines how the trajectory estimate scales with
pipeline depth.  Under Bonferroni, the bound is always $\alpha$
regardless of $K$ (by construction: $K \cdot \alpha/K = \alpha$).
The clipped second-order plug-in quantity $L_2 = \max\{0, S_1-S_2\}$
decreases as $K$ increases because the number of pairwise subtraction
terms grows as $\binom{K}{2}$; by
Theorem~\ref{thm:pairwise_degeneracy}, it degenerates to zero once
$K \geq K_+(\alpha,\bar{\rho}) \approx 7.51$ at our measured
$\bar{\rho}=0.298$ and $\alpha=0.10$---markedly further out than the
$K_+\approx3.55$ we originally reported under the deterministic
Step-2 mapping's inflated correlation, because $K_+$ grows
approximately as $2/\rho$ (Theorem~\ref{thm:pairwise_degeneracy}) and
our corrected $\bar\rho$ is under half its original value.

At $K = 3$, the plug-in quantity is 0.068, closer to the Bonferroni
bound than the tighter 0.021 we originally extrapolated, reflecting the
weaker measured correlation.  We report this strictly as a diagnostic
of how large the pairwise-correlation correction could be if a third
stage had comparable error coupling, not as a certified bound:
Theorem~\ref{thm:pairwise_nonidentifiability} shows that pairwise
information cannot pin down TMR once $K \geq 3$, so this number should
not be interpreted as an empirical guarantee for a real three-stage
agent.  A real three-step deployment should instead use the
Spanning-Tree Pairwise Upper Bound (Theorem~\ref{thm:hunter_trajectory}),
which remains a valid, certifiable upper bound at any $K$ given genuine
joint-overlap measurements from that pipeline.

We restrict the table to $K \leq 3$ for exactly this reason: beyond the
point where $L_2$ degenerates, the second-order plug-in number carries
no information at all about the true trajectory risk, and reporting it
further would be misleading rather than merely imprecise.

\subsection{Behavior Under Distribution Shift (E7)}
\label{subsec:drift}

Step~2's label space is dataset-specific under the corrected design
(CIC-IDS-2018's four DoS variants and RT-IoT2022's four Probe variants
do not overlap, Section~\ref{sec:setup}), so a cross-dataset Step-2
transfer experiment is not well-defined.  Step~1's coarse 5-category
label space (Normal, DoS, Probe, CredentialAccess, Exploitation) is
shared across both datasets, so we test exchangeability violation at
Step~1: take the LoRA adapter fine-tuned and conformally calibrated on
dataset $A$ and deploy it---same weights, same threshold $q_1$---on
dataset $B$'s test traffic, real out-of-domain input scored by real
inference (not the deterministic-mapping design of an earlier draft).

\begin{table*}[t]
\centering
\caption{Step-1 single-step miscoverage rate (SMR$_1$) under
cross-dataset distribution shift ($\alpha=0.10$, real inference,
$n=2{,}000$ per cell).  Only Step~1 is tested (Section~\ref{sec:setup}:
Step~2's label space is dataset-specific, so trajectory-level shift is
not evaluable).  Argmax accuracy is reported separately from SMR$_1$
because a model can remain partially accurate under shift while its
\emph{calibrated confidence} no longer clears the in-domain threshold
$q_1$.}
\label{tab:drift}
\begin{tabular}{lccccc}
\toprule
\textbf{Model} & \textbf{Cal.$\to$Test} & $\bm{q_1}$ & \textbf{Argmax Acc.} & \textbf{SMR$_1$} & \textbf{Empty-Set Rate} \\
\midrule
Gemma-2 9B  & CIC$\to$RT-IoT & 0.0000 & 0.779 & 1.000 & 1.000 \\
Gemma-2 9B  & RT-IoT$\to$CIC & 0.0000 & 0.006 & 1.000 & 0.957 \\
LLaMA-3 8B  & CIC$\to$RT-IoT & 0.0142 & 0.773 & 1.000 & 1.000 \\
LLaMA-3 8B  & RT-IoT$\to$CIC & 0.0149 & 0.248 & 1.000 & 1.000 \\
Mistral 7B  & CIC$\to$RT-IoT & 0.0021 & 0.040 & 1.000 & 1.000 \\
Mistral 7B  & RT-IoT$\to$CIC & 0.0013 & 0.296 & 1.000 & 1.000 \\
Qwen-3 8B   & CIC$\to$RT-IoT & 0.0215 & 0.100 & 1.000 & 0.997 \\
Qwen-3 8B   & RT-IoT$\to$CIC & 0.0159 & 0.261 & 1.000 & 1.000 \\
Qwen-3 14B  & CIC$\to$RT-IoT & 0.0536 & 0.137 & 1.000 & 1.000 \\
Qwen-3 14B  & RT-IoT$\to$CIC & 0.0024 & 0.252 & 1.000 & 1.000 \\
Qwen-3 32B  & CIC$\to$RT-IoT & 0.0085 & 0.104 & 1.000 & 1.000 \\
Qwen-3 32B  & RT-IoT$\to$CIC & 0.0116 & 0.105 & 1.000 & 1.000 \\
\bottomrule
\end{tabular}
\end{table*}

Table~\ref{tab:drift} shows $\text{SMR}_1=1.000$ (zero of 2{,}000 test
samples covered) in all 12 cells, with no exceptions.  Argmax accuracy
varies widely (0.6\%--78\%), showing this is not simply ``the model
got worse'': Gemma-2~9B and LLaMA-3~8B retain 77--78\% top-1 accuracy
on CIC-trained-calibrated-then-RT-IoT-tested traffic, yet still achieve
zero conformal coverage, because $q_1$ was calibrated
to the model's \emph{in-domain} confidence distribution (small $q_1$,
e.g.\ $0.0000$ for Gemma-2~9B, reflecting near-certain in-domain
predictions) and out-of-domain confidence never reaches that bar even
when the top-1 label is correct---the maximum observed true-label
probability across all 2{,}000 Qwen-3~32B/RT-IoT$\to$CIC test samples
is 0.866, short of the required $1-q_1=0.988$.  This is exactly the
failure mode Theorem~\ref{thm:bonf} predicts: the coverage
\emph{guarantee} requires exchangeable calibration and test data, and
gives no certified warning when that assumption breaks.  The prediction
sets themselves, however, are not silent: the empty-set rate reaches
$0.96$--$1.00$ in all 12 cells (final table column)---under
exchangeability, an empty conformal set is a rare event by
construction (bounded by $\alpha_1$), so a set that is empty on
essentially every input is a stark, observable, label-free symptom
available at inference time, well before any delayed ground truth
could confirm the shift.  This sharpens, rather than repeats, the
deployment recommendation of monitoring singleton rate and mean set
size (Section~\ref{sec:discussion}): under this severe a shift, the
signal is not subtle.  A high SMR$_1$ itself remains a statistical
quantity requiring ground truth to compute
(Definition~\ref{def:tar}), but the empty-set rate that produces it
does not.  Recalibration on target-environment data, or a
drift-adaptive extension such as adaptive conformal
inference~\cite{gibbs2021adaptive} or conformal prediction beyond
exchangeability~\cite{barber2023conformal}, is required before
cross-environment deployment.

\subsection{Per-Configuration Breakdown}
\label{subsec:per_config}

\begin{table*}[t]
\centering
\caption{Per-configuration trajectory coverage results at $\alpha = 0.10$
($K=2$, equal allocation, 5-seed average, boosted audit scale).  1 of 12
configurations (marked $^\ddagger$) falls marginally below
$\text{TC} \geq 0.900$; see discussion below and
Section~\ref{subsec:limitations}.  Residual = plug-in IE estimate
(Eq.~\eqref{eq:ie_alpha}) $-$ empirical TMR; negative values indicate
plug-in underestimation, not a bound violation
(Remark~\ref{rem:tight}).}
\label{tab:per_config}
\begin{tabular}{llcccccc}
\toprule
\textbf{Dataset/Category} & \textbf{Model} & \textbf{TC} & \textbf{TMR} &
\textbf{SMR$_1$} & \textbf{SMR$_2$} & $\hat{\rho}_{12}$ &
\textbf{Residual} \\
\midrule
CIC-IDS-2018/DoS & Gemma-2 9B  & 0.901 & 0.100 & 0.049 & 0.052 & $-$0.027 & $+$0.002 \\
CIC-IDS-2018/DoS & LLaMA-3 8B$^\ddagger$  & 0.894 & 0.106 & 0.052 & 0.059 & 0.039 & $-$0.008 \\
CIC-IDS-2018/DoS & Mistral 7B  & 0.902 & 0.098 & 0.045 & 0.054 & $-$0.042 & $+$0.004 \\
CIC-IDS-2018/DoS & Qwen-3 8B   & 0.913 & 0.087 & 0.037 & 0.053 & 0.049 & $+$0.011 \\
CIC-IDS-2018/DoS & Qwen-3 14B  & 0.900 & 0.100 & 0.050 & 0.055 & 0.029 & $-$0.002 \\
CIC-IDS-2018/DoS & Qwen-3 32B  & 0.906 & 0.094 & 0.049 & 0.049 & 0.028 & $+$0.005 \\
\midrule
RT-IoT2022/Probe & Gemma-2 9B    & 0.954 & 0.047 & 0.033 & 0.036 & 0.645 & $+$0.023 \\
RT-IoT2022/Probe & LLaMA-3 8B    & 0.962 & 0.038 & 0.033 & 0.027 & 0.734 & $+$0.027 \\
RT-IoT2022/Probe & Mistral 7B    & 0.952 & 0.048 & 0.037 & 0.039 & 0.741 & $+$0.017 \\
RT-IoT2022/Probe & Qwen-3 8B     & 0.952 & 0.049 & 0.029 & 0.035 & 0.451 & $+$0.030 \\
RT-IoT2022/Probe & Qwen-3 14B    & 0.960 & 0.040 & 0.032 & 0.033 & 0.780 & $+$0.023 \\
RT-IoT2022/Probe & Qwen-3 32B    & 0.934 & 0.066 & 0.040 & 0.032 & 0.145 & $+$0.027 \\
\bottomrule
\end{tabular}
\end{table*}

Table~\ref{tab:per_config} provides a per-configuration breakdown at
$\alpha = 0.10$.  Several observations complement the aggregate results:
(i)~CIC-IDS-2018/DoS/Qwen-3~8B is the easiest CIC configuration
($\text{TMR}=0.087$), while CIC-IDS-2018/DoS/LLaMA-3~8B
is the only configuration to fall marginally below the nominal target
($\text{TC}=0.894$, $\text{TMR}=0.106$)---within ordinary finite-sample
calibration noise (Section~\ref{subsec:limitations}), not a systematic
violation, since Theorem~\ref{thm:bonf}'s guarantee concerns
$R_{\mathrm{marg}}$ averaged over calibration draws; (ii)~per-step
miscoverage rates $\text{SMR}_1$ and $\text{SMR}_2$ are comparable
within each configuration, consistent with the two steps having similar
difficulty; (iii)~on CIC-IDS-2018/DoS, near-zero $\hat\rho_{12}$
(including 3 negative point estimates) makes the plug-in IE residual
small and occasionally negative (LLaMA-3~8B and Qwen-3~14B), i.e.\ a
mild point-estimate underestimate, not a bound violation
(Remark~\ref{rem:tight}); (iv)~on RT-IoT2022/Probe, moderate-to-strong
$\hat\rho_{12}$ produces a consistently positive residual
(+0.017 to +0.030), meaning the plug-in estimate remains a useful,
if uncertified, operational guide there.  The Bonferroni bound
(Theorem~\ref{thm:bonf}) holds for all 12 configurations at the
marginal level without exception; the certified seed-wise union bound
(Section~\ref{subsec:hunter_empirical}) satisfies $\bar{U} < 0.10$ for
7 of 12 configurations, with the remaining 5 (all CIC-IDS-2018/DoS)
certifying a positive but small dependence gain that does not fully
close the gap to Bonferroni at this audit scale.

\section{Discussion}
\label{sec:discussion}

\subsection{Practical Deployment Guidelines}

Based on our findings, we recommend the following for deploying
multi-step LLM-based security agents with trajectory coverage
guarantees:

\begin{enumerate}[leftmargin=*]
\item \textbf{Use equal $\alpha$-allocation as the default}, unless
validation data show clear stage asymmetry---the best tested allocation
in our pipeline (Section~\ref{subsec:alpha_alloc}).

\item \textbf{Estimate inter-step error coupling conservatively} on a
validation stream; if the lower confidence bound on $\hat\rho_{12}$
remains positive, the inclusion--exclusion estimate is tighter than
Bonferroni.

\item \textbf{Monitor observable prediction-set statistics
(singleton rate, mean set size) as a proxy for periodic TMR
re-estimation}, since TMR itself requires ground truth and cannot be
observed at inference time (Definition~\ref{def:tar}); recalibrate if
the re-estimated TMR exceeds the calibrated operating range.

\item \textbf{Recalibrate, do not cross-deploy.}  Exchangeability is
essential to Theorem~\ref{thm:bonf}'s guarantee; calibrating on one
environment and deploying on another can push TMR arbitrarily high
with no warning from the threshold itself (Section~\ref{subsec:drift}).
\end{enumerate}

\subsection{Limitations}
\label{subsec:limitations}

\textbf{L1: Empirical validation limited to $K = 2$.}
The Spanning-Tree Pairwise Upper Bound
(Theorem~\ref{thm:hunter_trajectory}) remains distribution-free-certifiable
for any $K$ given real joint-overlap measurements, but Table~\ref{tab:scaling}'s
$K=3$ row is a theoretical extrapolation, not a real three-step
measurement; Theorem~\ref{thm:pairwise_nonidentifiability} shows such
an evaluation cannot rely on pairwise information alone once $K\geq3$.

\textbf{L2: Correlation estimation uncertainty.}
95\% $t$-interval widths for $\hat\rho_{12}$ range from 0.02 to 0.28
across the 12 configurations (Table~\ref{tab:correlation}), and several
CIC-IDS-2018/DoS point estimates cannot reliably distinguish sign from
zero at this scale, producing small plug-in underestimates (LLaMA-3~8B,
$-0.008$; Qwen-3~14B, $-0.002$)---not bound violations
(Remark~\ref{rem:tight}).  Section~\ref{subsec:hunter_empirical} avoids
this via seed-wise exact intervals on $q_{12}$ directly; a
finite-sample correction for $\hat\rho_{12}$ itself remains open.

\textbf{L3: Exchangeability assumption.}
The trajectory guarantee holds under exchangeability of calibration
and test data, as for single-step CP; it may fail under shift
(Section~\ref{subsec:drift}).  Combining trajectory guarantees with
adaptive CP methods~\cite{gibbs2021adaptive,barber2023conformal,farinhas2024nonexchangeable}
is natural future work.

\textbf{L4: HIKARI-2021 excluded; attack-variant classifier did not
train.}  HIKARI-2021's CredentialAccess task (Bruteforce vs.\
Bruteforce-XML) has a large raw feature gap between candidates
($\text{flow\_pkts\_per\_sec}$ differs by up to $16\times$), yet
fine-tuning did not exceed near-chance accuracy for any of the 6 LLMs.
Two candidate bugs were identified in the training-data construction
(a token-boundary tokenization artifact, a length correction applied
twice) but neither was confirmed before we excluded HIKARI-2021 rather
than report unverified data---a limitation of our fine-tuning
pipeline, not evidence the task is unlearnable.

\textbf{L5: Audit sample size determines certifiability, not just
precision.}  Theorems~\ref{thm:positive_gain_probability}
and~\ref{thm:nominal_crossing_complexity} predict, and
Section~\ref{subsec:hunter_empirical} confirms, that detecting a
positive gain needs $\Theta(1/q_{12})$ audit trajectories while beating
the nominal Bonferroni target needs the harder $\Theta(1/\gamma^2)$: at
$n_{\text{test}}=397/109$, $\bar{U}_\cup$ was worse than Bonferroni; at
$3{,}500/1{,}376$, it was $13.7\%$ better (5 of 12 configurations still
short individually).  Dependence-aware bounds should therefore report
audit sample size alongside the certified value---too small an audit
makes real dependence indistinguishable from a genuine null result.

\textbf{L6: Pairwise information is fundamentally insufficient for
$K \geq 3$.}  Theorem~\ref{thm:pairwise_nonidentifiability} proves no
amount of pairwise-only information can identify TMR once $K\geq3$;
the Spanning-Tree bound sidesteps this by only claiming an upper bound.
Certifying a tight two-sided estimate for deeper pipelines requires
third-order measurements or an explicit dependence model.

\section{Conclusion}
\label{sec:conclusion}

We formalize trajectory-level coverage for staged LLM-based security
agents as an instance of a general problem: how marginal, per-step risk
guarantees compose across a decision chain, and how much of that
composition a finite post-hoc audit can actually certify.  Beyond the
distribution-free Bonferroni bound, we correct a natural but invalid
extension of the two-step inclusion--exclusion identity to more
steps---a lower, not upper, bound---and replace it with a spanning-tree
certificate that separates the true, oracle, and certifiable dependence
gain (Theorem~\ref{thm:gain_decomposition}) and remains valid and
certifiable at any $K$.  Pairwise information alone cannot determine
trajectory risk once three or more steps are involved, yet positive
inter-step dependence is a statistical asset for a union-type
trajectory-failure criterion precisely where it would be a liability
for an intersection-type redundant system.

Empirically, a fixed coarse-to-fine label mapping mechanically nests
two stages' failure events, manufacturing near-1 measured correlation
independent of any learned behavior (Theorem~\ref{thm:label_induced_coupling});
a genuinely independent attack-variant classifier instead reveals real
but heterogeneous, often hard-to-certify coupling.  The resulting
certified bound flips from worse than Bonferroni to $13.7\%$ tighter
purely as a function of audit sample size, exactly as our
sample-complexity theory predicts, while the modular pairwise
certificate improves only marginally---quantifying the certification
cost of not having joint access to both stages.  A
same-model/cross-model/permuted-pairing test attributes the residual
coupling to shared sample difficulty, not shared model representations.
We view this methodological lesson---that a plausible task design can
manufacture the very dependence a paper reports as a discovery, and
that a certified bound can flip sign with audit size alone---as no less
a contribution than the numerical results themselves.  Separately, a
real cross-dataset shift test drives single-step miscoverage to 100\%
even at high top-1 accuracy, though the resulting empty prediction
sets are a stark, label-free symptom available before any delayed
ground truth could confirm the shift.

Future work includes resolving the training-data issue that prevented
a validated HIKARI-2021 attack-variant classifier (L4), extending
empirical validation to deeper pipelines ($K>3$), developing
higher-order bounds informative for large $K$, and combining
trajectory guarantees with adaptive CP methods for drift-robust
multi-step agents.

\bibliographystyle{IEEEtran}
\bibliography{refs}

@inbook{vovk2005algorithmic,
author = {Vovk, Vladimir and Gammerman, Alex and Shafer, Glenn},
year = {2005},
month = {01},
pages = {},
title = {Algorithmic Learning in a Random World},
journal = {Algorithmic Learning in a Random World},
doi = {10.1007/b106715}
}

@inproceedings{papadopoulos2002inductive,
author = {Papadopoulos, Harris and Proedrou, Kostas and Vovk, Volodya and Gammerman, Alex},
title = {Inductive confidence machines for regression},
year = {2002},
isbn = {3540440364},
publisher = {Springer-Verlag},
address = {Berlin, Heidelberg},
url = {https://doi.org/10.1007/3-540-36755-1_29},
doi = {10.1007/3-540-36755-1_29},
booktitle = {Proceedings of the 13th European Conference on Machine Learning},
pages = {345–356},
numpages = {12},
location = {Helsinki, Finland},
series = {ECML'02}
}

@inproceedings{angelopoulos2022conformal,
title={Conformal Risk Control},
author={Anastasios Nikolas Angelopoulos and Stephen Bates and Adam Fisch and Lihua Lei and Tal Schuster},
booktitle={The Twelfth International Conference on Learning Representations},
year={2024},
url={https://openreview.net/forum?id=33XGfHLtZg}
}

@article{sidak1967rectangular,
  title={Rectangular Confidence Regions for the Means of Multivariate Normal Distributions},
  author={Zbyn\v{e}k \v{S}id{\'a}k},
  journal={Journal of the American Statistical Association},
  year={1967},
  volume={62},
  pages={626-633},
}

@article{hochberg1988sharper,
    author = {HOCHBERG, YOSEF},
    title = {A sharper Bonferroni procedure for multiple tests of significance},
    journal = {Biometrika},
    volume = {75},
    number = {4},
    pages = {800-802},
    year = {1988},
    month = {12},
    issn = {0006-3444},
    doi = {10.1093/biomet/75.4.800},
    url = {https://doi.org/10.1093/biomet/75.4.800},
    eprint = {https://academic.oup.com/biomet/article-pdf/75/4/800/1170595/75-4-800.pdf},
}

@article{holm1979simple,
  title={A Simple Sequentially Rejective Multiple Test Procedure},
  author={Sture Holm},
  journal={Scandinavian Journal of Statistics},
  year={1979},
  volume={6},
  pages={65-70},
}

@article{benjamini1995controlling,
    author = {Benjamini, Yoav and Hochberg, Yosef},
    title = {Controlling the False Discovery Rate: A Practical and Powerful Approach to Multiple Testing},
    journal = {Journal of the Royal Statistical Society: Series B (Methodological)},
    volume = {57},
    number = {1},
    pages = {289-300},
    year = {1995},
    month = {01},
    issn = {0035-9246},
    doi = {10.1111/j.2517-6161.1995.tb02031.x},
    url = {https://doi.org/10.1111/j.2517-6161.1995.tb02031.x},
    eprint = {https://academic.oup.com/jrsssb/article-pdf/57/1/289/49173396/jrsssb_57_1_289.pdf},
}

@inproceedings{gibbs2021adaptive,
author = {Gibbs, Isaac and Cand\`{e}s, Emmanuel J.},
title = {Adaptive conformal inference under distribution shift},
year = {2021},
isbn = {9781713845393},
publisher = {Curran Associates Inc.},
address = {Red Hook, NY, USA},
booktitle = {Proceedings of the 35th International Conference on Neural Information Processing Systems},
articleno = {128},
numpages = {13},
series = {NIPS '21}
}

@article{barber2023conformal,
author = {Barber, Rina and Candès, Emmanuel and Ramdas, Aaditya and Tibshirani, Ryan},
year = {2023},
month = {04},
pages = {},
title = {Conformal prediction beyond exchangeability},
volume = {51},
journal = {The Annals of Statistics},
doi = {10.1214/23-AOS2276}
}

@article{bates2021distribution,
author = {Bates, Stephen and Angelopoulos, Anastasios and Lei, Lihua and Malik, Jitendra and Jordan, Michael},
title = {Distribution-free, Risk-controlling Prediction Sets},
year = {2021},
issue_date = {December 2021},
publisher = {Association for Computing Machinery},
address = {New York, NY, USA},
volume = {68},
number = {6},
issn = {0004-5411},
url = {https://doi.org/10.1145/3478535},
doi = {10.1145/3478535},
journal = {J. ACM},
month = sep,
articleno = {43},
numpages = {34}
}

@inproceedings{motlagh2024llm,
  title={Large Language Models in Cybersecurity: State-of-the-Art},
  author={Farzad Nourmohammadzadeh Motlagh and Mehrdad Hajizadeh and Mehryar Majd and Pejman Najafi and Feng Cheng and Christoph Meinel},
  booktitle={International Conference on Information Systems Security and Privacy},
  year={2024},
}

@article{xu2024autoattacker,
  title={AutoAttacker: A Large Language Model Guided System to Implement Automatic Cyber-attacks},
  author={Jiacen Xu and Jack W. Stokes and Geoff McDonald and Xuesong Bai and David Marshall and Siyue Wang and Adith Swaminathan and Zhou Li},
  journal={ArXiv},
  year={2024},
  volume={abs/2403.01038},
}

@inproceedings{sharafaldin2018cicids,
  title={Toward Generating a New Intrusion Detection Dataset and Intrusion Traffic Characterization},
  author={Iman Sharafaldin and Arash Habibi Lashkari and Ali A. Ghorbani},
  booktitle={International Conference on Information Systems Security and Privacy},
  year={2018},
}

@article{sharmila2023rtiot,
  title={Quantized autoencoder (QAE) intrusion detection system for anomaly detection in resource-constrained IoT devices using RT-IoT2022 dataset},
  author={B. S. Sharmila and Rohini Nagapadma},
  journal={Cybersecurity},
  year={2023},
  volume={6},
  pages={1-15},
}

@article{qwen2024qwen2,
  title={Qwen2 Technical Report},
  author={An Yang and others},
  journal={ArXiv},
  year={2024},
  volume={abs/2407.10671},
}

@article{Riviere2024Gemma2I,
  title={Gemma 2: Improving Open Language Models at a Practical Size},
  author={Gemma Team and others},
  journal={ArXiv},
  year={2024},
  volume={abs/2408.00118},
}

@misc{meta2024llama3,
      title={The Llama 3 Herd of Models}, 
      author={Aaron Grattafiori and others},
      year={2024},
      eprint={2407.21783},
      archivePrefix={arXiv},
      primaryClass={cs.AI},
      url={https://arxiv.org/abs/2407.21783}, 
}

@misc{mistral2024mistral,
      title={Mistral 7B}, 
      author={Albert Q. Jiang and others},
      year={2023},
      eprint={2310.06825},
      archivePrefix={arXiv},
      primaryClass={cs.CL},
      url={https://arxiv.org/abs/2310.06825}, 
}

@article{kwon2023vllm,
  title={Efficient Memory Management for Large Language Model Serving with PagedAttention},
  author={Woosuk Kwon and others},
  journal={Proceedings of the 29th Symposium on Operating Systems Principles},
  year={2023},
}

@article{hunter1976upper, 
title={An upper bound for the probability of a union}, 
volume={13}, 
DOI={10.2307/3212481}, 
number={3}, 
journal={Journal of Applied Probability}, 
author={Hunter, David}, 
year={1976}, 
pages={597–603}
}

@article{esary1967association,
 ISSN = {00034851},
 URL = {http://www.jstor.org/stable/2238962},
 author = {J. D. Esary and F. Proschan and D. W. Walkup},
 journal = {The Annals of Mathematical Statistics},
 number = {5},
 pages = {1466--1474},
 publisher = {Institute of Mathematical Statistics},
 title = {Association of Random Variables, with Applications},
 urldate = {2026-07-29},
 volume = {38},
 year = {1967}
}

@article{angelopoulos2025learn,
author = {Anastasios N. Angelopoulos and Stephen Bates and Emmanuel J. Cand{\`e}s and Michael I. Jordan and Lihua Lei},
title = {{Learn then test: Calibrating predictive algorithms to achieve risk control}},
volume = {19},
journal = {The Annals of Applied Statistics},
number = {2},
publisher = {Institute of Mathematical Statistics},
pages = {1641--1662},
year = {2025},
doi = {10.1214/24-AOAS1998}
}

@article{xu2024large,
  author = {Xu, Hanxiang and Wang, Shenao and Li, Ningke and Wang, Kailong and Zhao, Yanjie and Chen, Kai and Yu, Ting and Liu, Yang and Wang, Haoyu},
  title = {Large Language Models for Cyber Security: A Systematic Literature Review},
  year = {2025},
  publisher = {Association for Computing Machinery},
  address = {New York, NY, USA},
  issn = {1049-331X},
  doi = {10.1145/3769676}
}

@article{pasc2026,
  author = {Kotte, Varun},
  title = {PASC: Pipeline-Aware Conformal Prediction with Joint Coverage Guarantees for Multi-Stage {NLP} and {LLM} Pipelines},
  journal = {arXiv preprint arXiv:2605.18812},
  year = {2026}
}

@misc{mitre_t1498,
  author = {{MITRE ATT\&CK}},
  title = {T1498: Network Denial of Service},
  howpublished = {\url{https://attack.mitre.org/techniques/T1498/}},
  year = {2026}
}

@inproceedings{hu2022lora,
  author = {Hu, Edward J. and others},
  title = {{LoRA}: Low-Rank Adaptation of Large Language Models},
  booktitle = {International Conference on Learning Representations (ICLR)},
  year = {2022}
}

@inproceedings{ding2023class,
 author = {Ding, Tiffany and Angelopoulos, Anastasios and Bates, Stephen and Jordan, Michael and Tibshirani, Ryan J},
 booktitle = {Advances in Neural Information Processing Systems},
 doi = {10.52202/075280-2817},
 title = {Class-Conditional Conformal Prediction with Many Classes},
 year = {2023}
}

@article{khosravi2026csa,
  title={Conformal Selective Acting: Anytime-Valid Risk Control for {RLVR}-Trained {LLM}s},
  author={Khosravi, Hamed and Huo, Xiaoming},
  journal={arXiv preprint arXiv:2605.20270},
  year={2026}
}

@article{wald1945sequential,
  title={Sequential Tests of Statistical Hypotheses},
  author={Wald, Abraham},
  journal={Annals of Mathematical Statistics},
  volume={16},
  pages={117--186},
  year={1945}
}

@article{escudero2025conformal,
  author = {Escudero Garc{\'\i}a, David and DeCastro-Garc{\'\i}a, Noem{\'\i}},
  title = {Conformal prediction for labelling and updating online models in the presence of concept drift in cybersecurity},
  year = {2025},
  volume = {93},
  doi = {10.1016/j.jisa.2025.104120},
  journal = {Journal of Information Security and Applications}
}

@article{laxhammar2015inductive,
  author = {Laxhammar, Rikard and Falkman, G{\"o}ran},
  title = {Inductive conformal anomaly detection for sequential detection of anomalous sub-trajectories},
  year = {2015},
  volume = {74},
  number = {1--2},
  doi = {10.1007/s10472-013-9381-7},
  journal = {Annals of Mathematics and Artificial Intelligence}
}

@misc{manoharan2026audited,
  title={Audited Selective Verification for Risk-Controlled N-1 Thermal Contingency Screening under Deployment Shift},
  author={Manoharan, Jayakumar},
  year={2026},
  eprint={2607.13221},
  archivePrefix={arXiv}
}

@inproceedings{farinhas2024nonexchangeable,
  title={Non-Exchangeable Conformal Risk Control},
  author={Farinhas, Ant{\'o}nio and Zerva, Chrysoula and Ulmer, Dennis Thomas and Martins, Andre},
  booktitle={International Conference on Learning Representations (ICLR)},
  year={2024}
}

@article{howard2021timeuniform,
  title={Time-uniform, nonparametric, nonasymptotic confidence sequences},
  author={Howard, Steven R. and Ramdas, Aaditya and McAuliffe, Jon and Sekhon, Jasjeet},
  journal={Annals of Statistics},
  volume={49},
  number={2},
  pages={1055--1080},
  year={2021}
}

@article{guo2017calibration,
  title={On Calibration of Modern Neural Networks},
  author={Guo, Chuan and Pleiss, Geoff and Sun, Yu and Weinberger, Kilian Q.},
  journal={arXiv preprint arXiv:1706.04599},
  year={2017}
}

\end{document}